\documentclass[journal,10pt]{IEEEtran}
\usepackage{mathrsfs}
\usepackage{bbding}
\usepackage{amsfonts}
\usepackage{amsmath}
\usepackage{graphicx}
\usepackage{subfigure}
\usepackage{latexsym}
\usepackage{amssymb}
\usepackage{amsmath}
\usepackage{enumerate}
\usepackage{color}
\usepackage{mdframed}
\usepackage{boxedminipage}

\usepackage{color}
\usepackage{calc}
\usepackage{amsmath}
\usepackage{amssymb}
\usepackage{graphicx}
\usepackage{esint}

\newtheorem{Assumption}{Assumption}

\newtheorem{Definition}{Definition}
\newtheorem{Lemma}{Lemma}
\newtheorem{Problem}{Problem}
\newtheorem{Remark}{Remark}
\newtheorem{Theorem}{Theorem}
\newtheorem{Corollary}{Corollary}

\usepackage[font=small]{caption}
\newtheorem{Classification}{Classification}
\newcommand{\txtblue}{\textcolor{black}}
\newcommand{\tblue}{\textcolor{black}}

\newcommand{\KO}{~\hfill\IEEEQED}
\newcommand{\bs}{\begin{small}}
\newcommand{\es}{\end{small}}

\allowdisplaybreaks[4]

\begin{document}
\title{Closed-Form Delay-Optimal Power Control \\ for Energy Harvesting Wireless System \\ with Finite Energy Storage}

\author
{\IEEEauthorblockN{Fan Zhang, \emph{Student Member, IEEE}, Vincent K. N. Lau, \emph{Fellow, IEEE}} \\ Department of ECE,  Hong Kong University of Science and Technology, Hong Kong \\ Email: fzhangee@ust.hk, eeknlau@ust.hk}
\maketitle

 \vspace{-1cm}
\begin{abstract}
In this paper, we consider  delay-optimal  power control for an energy harvesting wireless system with finite energy storage.  The wireless system is powered solely by a renewable energy source with bursty data arrivals, and is characterized by a \emph{data queue} and an \emph{energy queue}. We consider a delay-optimal power control problem and formulate  an infinite horizon average cost Markov Decision Process (MDP). To deal with the curse of dimensionality, we introduce a virtual continuous time system and  derive closed-form approximate priority functions for the discrete time MDP at various operating regimes.  Based on the approximation,  we obtain an online power control solution which is adaptive  to the channel state information as well as the data and energy queue state information.   The derived power control solution has a \emph{multi-level water-filling} structure, where the \emph{water level} is determined jointly by the data and energy queue  lengths. We show through simulations that the proposed scheme has significant performance gain compared with various baselines.
\end{abstract}
 \vspace{-0.5cm}
\section{Introduction}
Recently, green communication has received considerable attention since it will play an important role in enhancing energy efficiency and reducing carbon emissions in future wireless networks \cite{greencom1}, \cite{greencom2}.  To support green communication, \emph{energy harvesting} techniques such as solar panels, wind turbines and thermoelectric generators  \cite{renewablesource} have become popular for enabling the transmission nodes to harvest energy from the environment.  While the renewable energy sources may appear to be virtually free and they are random in nature, energy storage is needed to buffer the unstable supply of the renewable energy \cite{huanghenergy}.  In \cite{phyener1} and \cite{phyener2}, the  authors propose transmission policies that minimize the transmission time for a given amount of data in point-to-point and broadcast  energy harvesting networks  with an infinite capacity battery. However, the infinity capacity battery assumption is not realistic in practice.  In \cite{phyener3} and \cite{phyener4}, offline power allocation policies are proposed by solving short-term throughput maximization problems under finite energy storage capacity  in a finite time horizon.  However, the above works \cite{phyener1}--\cite{phyener4} assume that the realizations of the energy arrival processes are known in advance (i.e., non-causal knowledge of future arrivals). Furthermore,  the above  proposed  policies  \cite{phyener1}--\cite{phyener4} are based on the assumption    that there are infinite data backlogs at the transmitters so that  the applications are delay-insensitive. In practice, it is very important to consider bursty data arrivals, bursty energy arrivals  and delay requirements in  designing the power control policy for  delay-sensitive applications. 

In this paper, we are interested in the online power control solution in a wireless system powered by a renewable energy source to support real-time delay-sensitive  applications. The wireless transmitter is powered solely by an energy harvesting storage with limited energy storage capacity. Unlike the previous proposed schemes  \cite{phyener1}--\cite{phyener4}, we consider an online control  policy, in the sense that we only have causal knowledge of the system states. Specifically, to support real-time applications with bursty data arrivals and bursty renewable energy arrivals, it is very important to dynamically  control the transmit power  that is adaptive to the channel state information (CSI), the data queue length (DQSI) and the energy queue length (EQSI). The CSI reveals the {\em transmission opportunities} of the time-varying physical channels. The DQSI reveals the {\em urgency of the data flows} and the EQSI reveals {\em the availability of the renewable energy}. It is highly non-trivial to strike a good balance between these factors.

Online power control  adaptive to the CSI, the DQSI and the EQSI is quite challenging because the associated optimization problem belongs to an infinite-dimensional stochastic optimization problem. There is intense research interest in exploiting renewable energy in communication network designs. In  \cite{lya}, the authors use large deviations theory to find the closed-form expression for the buffer overflow probability and design an energy-efficient scheme for maximizing the decay exponent of this probability. \tblue{In \cite{queuestable}, the authors propose   throughput-optimal control policies (in the stability sense) that are adaptive to the CSI, the DQSI and  the EQSI for a point-to-point energy harvesting network.}  In \cite{stab1} and \cite{stab2}, the authors extend the Lyapunov optimization framework to derive  energy management algorithms,  which can stabilize the data queue for energy harvesting networks with finite energy storage capacity. Note that the buffer overflow probability  and the queue stability are weak forms of delay performance, and it is of great importance to study the control policies that minimize the average delay of the queueing network.  A systematic approach in dealing with the delay-optimal  control is to formulate the problem into an Markov Decision Process (MDP) \cite{mdp1}, \cite{mdp2}. In \cite{submdp}, the authors propose several heuristic event-based adaptive transmission policies on the basis of a finite horizon MDP formulation. These solutions are suboptimal and with no performance guarantee. In \cite{huanghenergy}, the authors consider online power control for  the interference network with a renewable energy supply by solving an infinite horizon average cost MDP. \tblue{The authors in \cite{queuestable} also propose an online delay-optimal power control policy by solving an infinite horizon MDP for energy harvesting networks.}   However, the MDP problems in \cite{huanghenergy} and  \cite{queuestable} are solved using numerical iteration algorithms, such as value iteration or policy iteration algorithms (\emph{Chap. 4 in Vol. 2} of  \cite{mdp2}), which suffer from slow convergence and a lack of insight.  There are some existing works that adopt MDP/POMDP approaches to solve the stochastic resource allocation problems for energy harvesting wireless sensor networks \tblue{\cite{new1}--\cite{ref1}}. In \cite{new1}, the authors consider a simple birth-death model for the energy queue dynamics and obtain a threshold-like  data transmission scheme by maximizing an average data rate reward using the MDP approach. However, the energy queue model  considered in \cite{new1} is a simplified model and the approach therein  cannot be applied in our scenario with general energy queue dynamics. In \cite{new2}--\tblue{\cite{ref2}}, the authors propose  an efficient power control scheme to minimize the average power consumption and the packet error rate. \tblue{In \cite{ref3} and \cite{ref4}, the authors consider on-off control of the sensor to maximize the event detection efficiency or maximize the discounted weighted sum transmitted data. However, the power control actions in \cite{new2}--\cite{ref4} are chosen  from  discrete and finite action spaces.} Hence,  the approaches in \cite{new2}--\tblue{\cite{ref4}} cannot be applied to our scenario where the power control action is chosen from a continuous action space. In \cite{new4} and \cite{new5}, the authors propose a power allocation scheme for an energy harvesting sensor network with finite energy buffer capacity by solving a POMDP problem. However, they consider non-causal control, which means that the realizations of the energy arrival processes are known in advance.  \tblue{In \cite{ref1}, the authors consider general energy queue model with online causal power control schemes. However,  the stochastic MDP/POMDP   problems in \cite{new1}--\cite{ref1} are solved using  numerical value iteration or policy iteration algorithms \cite{mdp2}.} In this paper, we focus on deriving a closed-form delay-optimal online power control solution that is adaptive to the CSI, the DQSI and the EQSI. There are several first order technical challenges associated with the stochastic optimization.

\begin{itemize}
	\item	\textbf{Challenges due to the Queue-Dependent Control:}
	In order to maintain low average delay performance and efficiently use the renewable energy in a finite capacity storage, it is important to dynamically control the transmit power  based on the CSI, the DQSI and the EQSI. As a result, the underlying problem embraces both information theory (to model the physical layer dynamics) and the queueing theory (to model the data  and energy queue dynamics) and is an infinite horizon stochastic optimization \cite{mdp1}, \cite{mdp2}. Such problems are well-known to be very challenging due to the infinite-dimensional optimization (w.r.t.  control policy) and lack of closed-form characterization of the value function in the optimality equation (i.e., the Bellman equation).
	\item	\textbf{Complex Coupling between the Data Queue and the Energy Queue:}
	The service rate of the transmitter in the energy harvesting network depends on the current available energy stored in the energy queue buffer. As such, the dynamics of the data queue and the energy queue are coupled together.  The associated stochastic optimization problem is a multi-dimensional MDP \cite{delaysurvey}. To solve the associated  Bellman equation, numerical brute-force approaches (e.g., value iteration and policy iteration \cite{mdp2}) can be adopted, but they are \txtblue{not practical}  and provide no design insights.  Therefore, it is desirable  to obtain a low complexity and insightful solution for the dynamic power control in the energy harvesting system.
	\item	\textbf{Challenges due to the Finite  Energy Storage   and Non-i.i.d. Energy Arrivals:}
	In practice, the energy storage (or battery) at the transmitter has finite capacity only. The finite renewable energy storage limit   induces a difficult {\em energy availability constraint} (the energy consumption per time slot cannot exceed the available energy  in the storage) in the stochastic optimization problem.   Furthermore, in the previous literature (e.g., \cite{huanghenergy}--\cite{lya}), the bursty energy arrivals are modeled as an i.i.d. process for analytical tractability. \tblue{In \cite{queuestable}, the authors also consider  periodic stationary energy arrivals for designing the power control policies.} In practice, most of the renewable energy arrivals are not i.i.d.. Such non-iid nature will have a huge impact on the dimensioning of the battery capacity.   	
\end{itemize}

In this paper, we  model the delay-optimal  power control problem as an infinite horizon average cost MDP. Specifically,  the stochastic MDP problem is to minimize the average delay  of the transmitter  subject to the energy availability constraint. By exploiting the special structure in our problem, we derive an equivalent Bellman equation to solve the MDP. We then introduce a virtual continuous time system (VCTS) where the evolutions of the  data  and energy queues are characterized by two coupled differential equations with reflections. We show that the priority function of the associated total cost problem in the VCTS is asymptotically optimal to that  of the discrete time MDP problem when the slot duration is sufficiently small. Using the  priority function in the VCTS as an approximation to the optimal priority function, we  derive  online power control solutions and obtain design insights from the structural properties of the priority function under different asymptotic regimes. The power control solution has a {\em multi-level water-filling} structure, where the DQSI and the EQSI determine the water level via the priority function.  Finally, we compare the proposed solution with various baselines  and show that significant performance gain can be achieved.

\begin{figure}[t]
  \centering
  \includegraphics[width=3.5in]{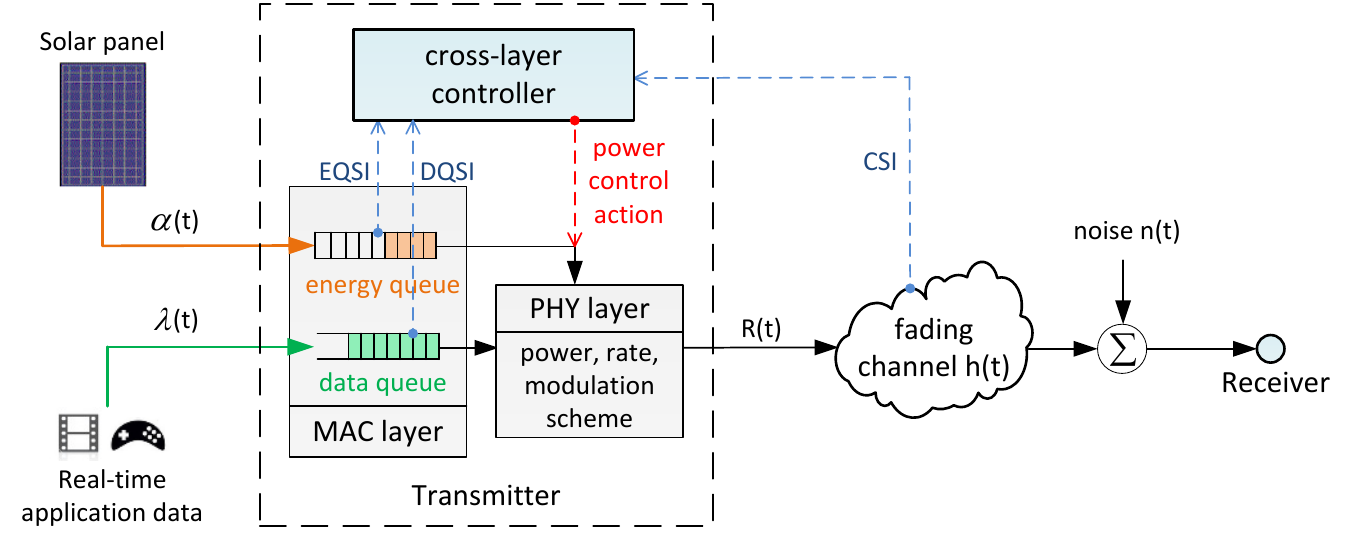}
  \caption{System model of the point-to-point energy harvesting system.}
  \label{topo}\vspace{-0.5cm}
\end{figure}

\section{System Model}
We  consider a point-to-point energy harvesting system with finite energy storage. Fig.~\ref{topo} illustrates the top-level system model, where the  transmitter is powered solely by the  energy harvesting storage with limited energy storage capacity. The transmitter acts as a \emph{cross-layer controller}, which takes the CSI, the DQSI and the EQSI as input and generates power control action as output. In this paper, the time dimension is partitioned into decision slot  \txtblue{indexed} by \txtblue{$n$ ($n=0,1,2,\dots$)} with duration $\tau$.   In the following subsections, we  elaborate on the physical layer model and the bursty data arrival model, as well as the  renewable energy arrival model.

\subsection{Physical Layer Model}
We consider a  point-to-point system as shown  in Fig.~\ref{topo}. The transmitter sends information to the receiver. Let $s$ be the transmitted information symbol and  the received signal  is given by
\begin{align}
	y = h \sqrt{p} s + z
\end{align}
where $h\in \mathbb{C}$ is the complex channel fading coefficient between the transmitter and the receiver, $p$ is the transmit power, and $z\sim \mathcal{CN}\left(0,1 \right)$ is the i.i.d. complex Gaussian additive channel noise. We have the following assumption on the channel model.
\begin{Assumption}	[Channel Model]	\label{CSIassum}	
	$h\left(\txtblue{n}\right)$ remains constant within each decision slot  and is  i.i.d. over  the slots. Specifically, we assume that $h\left(\txtblue{n}\right)$ follows a complex Gaussian distribution with zero mean and unit variance, i.e., $h\left( \txtblue{n}\right) \sim \mathcal{CN}\left(0,1 \right)$.\KO
\end{Assumption}

For given CSI realization $h$ and power control action $p$,  the achievable data rate (bit/s/Hz) for the transmitter-receiver pair is given by 
\begin{equation}		\label{rate1}
	R\left(h, p  \right) =  \log \left( 1+\txtblue{\zeta} p \left|h\right|^2 \right) 
\end{equation}
\txtblue{where $\zeta \in (0, 1]$ is a constant that is determined by the modulation and coding scheme (MCS) used in the system. For example, $\zeta=0.5$ for QAM constellation at BER= 1\% \cite{constellation} and $\zeta = 1$ for capacity-achieving coding (in which (\ref{rate1}) corresponds to the instantaneous mutual information). In this paper, our derived results are based on $\zeta=1$ for  simplicity, which can be easily extended to other MCS cases.}

\subsection{Bursty Data Source Model and Data Queue Dynamics}
As illustrated in Fig.~\ref{topo}, the  transmitter maintains a data queue for the bursty traffic flow towards the  receiver. Let  $\lambda\left(\txtblue{n}\right)\tau$ be the random new data arrival (bits)  at the end of the $\txtblue{n}$-th  decision slot at the transmitter. We have the following assumption on the data arrival process.
\begin{Assumption} [Bursty Data Source Model]	\label{assumeA}
	The data arrival process $\lambda \left(\txtblue{n}\right)$  is i.i.d. over the slots according to a general distribution $\Pr[\lambda]$ with finite average arrival rate $\mathbb{E}\left[\lambda \right]=\overline{\lambda}$.\KO
\end{Assumption}

Let $Q\left( \txtblue{n}\right) \in \mathcal{Q}$  denote the DQSI (bits) at the data queue of the transmitter at the beginning of the $\txtblue{n}$-th slot, where $\mathcal{Q}=[0, \infty)$ is the DQSI state space.  We assume that the transmitter is causal in the sense that  new data arrivals are observed after the  control actions are performed at each decision slot. Hence, the data queue dynamics  is given by
\begin{align}		\label{dataQ}
	Q\left(\txtblue{n}+1 \right) = \left[ Q\left(\txtblue{n}\right) -  R\left(h \left(\txtblue{n} \right), p \left(\txtblue{n}\right)  \right) \tau  \right]^+ + \lambda \left(\txtblue{n}\right) \tau
\end{align}	
where $x^+ \triangleq \max \left\{x, 0 \right\}$. \tblue{Note that $p(n)$ is transmit power of the transmitter  at time slot $n$ and the power  solely comes from the renewable energy source.} We shall define the renewable energy source model in the next subsection.

\subsection{Renewable Energy Source Model and Energy Queue Dynamics}
The power of the transmitter solely comes from the renewable energy source.  Specifically, the transmitter is  capable of harvesting energy from the environment, e.g., using solar panels, wind turbines and thermoelectric generators  \cite{renewablesource}. We assume that  the energy arrival process is block i.i.d. with block size $N$. The block i.i.d. energy arrival model  is used to take into account that the energy arrival process evolves at a different timescale w.r.t. that  of the data arrival process.   Let $\alpha\left(\txtblue{n}\right) \tau$ be the random renewable energy arrival (Joules) at the end of the $\txtblue{n}$-th  decision slot at the transmitter. We have the following assumption on the energy arrival process.

\begin{Assumption} [Block i.i.d. Renewable Energy Source Model]	\label{assumeAr}
	The energy arrival process $\alpha \left(n\right)$   is block  i.i.d. in the sense that $\alpha (\txtblue{n})$ is  constant\footnote{Specifically, $\alpha(\txtblue{n})$ is  constant when $kN\leq t <(k+1)N$ for any given $t$, where $k$ is a  positive integer.} for a block of $N$ slots and is i.i.d. between blocks  according to a general distribution $\Pr[\alpha]$ with finite average energy arrival rate $\mathbb{E}\left[\alpha \right]=\overline{\alpha}$.\KO
\end{Assumption}

Due to the random nature of the renewable energy, there is limited  energy storage capacity at the transmitter  to buffer the renewable energy arrivals.  Let $E\left(\txtblue{n} \right) \in \mathcal{E}$  denote the EQSI (Joules) at the beginning of the $\txtblue{n}$-th slot, where $\mathcal{E}=\left[0, N_E\right]$ is the EQSI state space and $N_E$ denotes the energy queue buffer size (i.e., energy storage capacity in Joules). 
\txtblue{\begin{Remark}	[Discussions  on the Finite Energy Queue Capacity]
	High-capacity renewable energy storage is very expensive \cite{r2ref2} and energy storage is one key cost component in renewable energy systems. As such, it is very important to consider the impact on how the finite renewable energy buffer affects the system performance. The analysis also serves as the first order dimensioning on how large an energy buffer is needed. ~\hfill~\IEEEQED
\end{Remark}}

Note that when the energy buffer is full, i.e., $E(\txtblue{n})=N_E$, additional energy cannot be harvested. Similarly, we assume that the transmitter  is causal so that the renewable energy arrival $E\left(\txtblue{n}\right)$ is observed only after the power actions. Hence, the  energy queue dynamics at the transmitter is given by
\begin{align}		\label{energyQ}
	E\left(\txtblue{n}+1 \right) = \min \big\{ E\left(\txtblue{n}\right) -   p \left(\txtblue{n} \right)   \tau + \alpha\left(\txtblue{n}\right)  \tau, N_E  \big\}
\end{align}	
where the renewable power consumption $p\left(\txtblue{n}\right)$ must  satisfy the following \emph{energy availability constraint}:	
\begin{align}	\label{eneavlcon}
	 p \left(\txtblue{n}  \right)   \tau \leq E(\txtblue{n} )
\end{align}
The energy  availability constraint means that the energy consumption at each time slot  cannot exceed the current available energy in the energy storage. Due to this constraint, the energy queue $E(\txtblue{n})$ in (\ref{energyQ}) will not go below zero (i.e., $E(\txtblue{n})\geq 0$ for all $\txtblue{n}$).  
\begin{Remark} \emph{(Coupling Property of  Data Queue  and Energy  Queue)} 	\label{coup_rem}
	 The   data queue dynamics  in (\ref{dataQ}) and the energy queue dynamics  in (\ref{energyQ}) are coupled together.  Specifically, the service rate $R\left(\txtblue{n} \right)$ in the data queue depends on the power control action $p\left(\txtblue{n} \right)$, which solely comes  from the energy queue buffer.\KO
\end{Remark}

\section{Delay-Optimal Problem Formulation}

In this section, we  formally define  the  power  control policy  and formulate the delay-optimal control problem for the point-to-point energy harvesting system.
\subsection{Power  Control Policy}
For notation convenience, denote $\boldsymbol{\chi}(n)=\left(h(n), Q(n), E(n)\right)$. Let $\mathcal{F}(n)=\sigma\big(\big\{\boldsymbol{\chi}(i):0\leq i\leq n\big\}\big)$ be the minimal $\sigma$-algebra containing the set $\big\{\boldsymbol{\chi}(i):0\leq i\leq n\big\}$, and $\big\{\mathcal{F}(n)\big\}$ be the associated filtration \cite{probabilityfilt}. At the beginning of  the $n$-th slot, the transmitter  determines the  power  control action based on the   following   control policy:
\begin{Definition} [Power  Control Policy]	\label{deff1}
	A   power   control policy for the transmitter  $\Omega$ is $\mathcal{F}(n)$-adapted at each time slot $n$, meaning that  the  power    control action $p(n)$ is adaptive to all the information $\boldsymbol{\chi}(i)$ up to tome $n$ (i.e., $\big\{\boldsymbol{\chi}(i):0\leq i\leq n\big\}$). Furthermore, the power control policy $\Omega$ satisfies the energy availability constraint in (\ref{eneavlcon}), i.e., $p \left(n\right)\tau \leq E(n)$ ($\forall n$).~\hfill\IEEEQED
\end{Definition}

Given a  control policy $\Omega$, the  random process $\left\{\boldsymbol{\chi}\left(n \right)\right\}$ is a controlled Markov chain with the following transition probability:
\begin{align}	
	  &\Pr\big[ \boldsymbol{\chi}\left(n+1 \right) \big| \boldsymbol{\chi}\left(n \right),  \Omega\big(\boldsymbol{\chi}\left(n \right) \big)\big] \notag \\
	 =& \Pr \big[h\left(n+1 \right) \big] \Pr \big[ Q \left(n+1\right) \big| Q\left(n \right), h(n),\Omega\big(\boldsymbol{\chi}\left(n \right) \big) \big] \notag \\
	 &\cdot \Pr \big[ E \left(n+1\right) \big| E\left(n \right),\Omega\big(\boldsymbol{\chi}\left(n \right) \big) \big]\label{trankernel}
\end{align}
where $ \Pr \big[ Q \left(n+1\right) \big| Q\left(n \right), h(n),\Omega\big(\boldsymbol{\chi}\left(n \right) \big) \big]$ is the data queue transition probability which is given by
\begin{align}	
	 &\Pr \big[ Q \left(n+1\right)=Q' \big| Q\left(n \right)=Q, h(n)=h,\Omega\big(\boldsymbol{\chi}\left(n \right) \big)=p \big]\notag \\
	 =&\left\{\begin{aligned}
	 &\Pr\left[\lambda(n)\right],\hspace{1cm}\text{if }Q'= \left[ Q -  R\left(h, p   \right) \tau  \right]^+ +\lambda(n)\tau\\
	 &0, \hspace{2.3cm}\text{otherwise} 
	 \end{aligned}\right.
  \end{align}
and  $ \Pr \big[ E \left(n+1\right) \big| E\left(n \right),\Omega\big(\boldsymbol{\chi}\left(n \right) \big) \big]$ is the energy queue transition probability which is given by
\begin{align}	
	 &\Pr \big[ E \left(n+1\right)=E' \big| E\left(n \right)=E,\Omega\big(\boldsymbol{\chi}\left(n \right) \big)=p \big]\notag \\
	 =&\left\{\begin{aligned}
	 &\Pr\left[\alpha\left(n\right)  \right],\hspace{1cm}\text{if }E'= \min \big\{ E -   p\tau + \alpha\left(n\right)  \tau, N_E  \big\}\\
	 &0, \hspace{2.4cm}\text{otherwise} 
	 \end{aligned}\right.
  \end{align}
Furthermore, we have the following definition on the admissible  control policy:
\begin{Definition}	[{Admissible Control Policy}]	\label{adddtdomain}
	{A policy $\Omega$ is  admissible if the following requirements are satisfied:}
	\begin{itemize}
		\item $\Omega$ is a unichain policy, i.e., the controlled Markov chain $\left\{\boldsymbol{\chi}\left(n \right)\right\}$ under $\Omega$ has a single recurrent class (and possibly some transient states) \cite{mdp2}.
		\item {The queueing system under $\Omega$ is stable in the sense that   $\lim_{n \rightarrow \infty} \mathbb{E}^{\Omega} \big[  Q^2(n)  \big] < \infty$, where $\mathbb{E}^{\Omega} $ means taking expectation w.r.t. the probability measure induced by the control policy $\Omega$. }~\hfill\IEEEQED
	\end{itemize}
\end{Definition}

\subsection{Problem Formulation}
As a result, under {an admissible}  control policy $\Omega$, the average delay cost of the energy harvesting system starting from a given initial  state $\boldsymbol{\chi}\left(0\right)$ is given by
\begin{align}	\label{delay_cost}
	\overline{D}\left(\Omega\right)  = \limsup_{N \rightarrow \infty} \frac{1}{N} \sum_{n=0}^{N-1} \mathbb{E}^{\Omega} \left[\frac{Q\left(n\right)}{\overline{\lambda}} \right]	
\end{align}	

We consider the following delay-optimal power control optimization for the energy harvesting system:
\begin{Problem}	[Delay-Optimal Power Control Optimization]  \label{IHAC_MDP}
	\begin{eqnarray} 
	\underset{\Omega}{\min} 	\ \overline{D}\left(\Omega\right)
\end{eqnarray}
where $\Omega$ satisfies the energy availability constraint according to Definition \ref{deff1}.~\hfill\IEEEQED
\end{Problem}

\subsection{Optimality Conditions}
While the MDP in Problem \ref{IHAC_MDP} is  difficult in general, we utilize  the i.i.d. assumption of the CSI to derive an \emph{equivalent optimality equation} as summarized below.
\begin{Theorem} [Sufficient Conditions for Optimality]	\label{LemBel}
	Assume there exists a ($\theta^\ast, \{ V^\ast\left(Q,E \right) \}$) that solves the following \emph{equivalent optimality equation}:
	\begin{align}	\label{OrgBel}
		 &\theta^\ast  + V^\ast \left(Q,E  \right)  \qquad \forall Q,E 	 \\
		 =&  \mathbb{E} \bigg[\min_{p<E/\tau}\Big[ \frac{Q}{\overline{\lambda}}  +  \sum_{Q',E'}\Pr \big[ Q',E '\big| \boldsymbol{\chi}, p\big]V^\ast \left(Q',E'\right) \Big]   \bigg| Q,E  \bigg] \notag
	\end{align}
	Furthermore,align all  admissible control policy $\Omega$ and initial queue state $\left(Q \left(0 \right),E \left(0 \right)\right)$, $V^\ast$ satisfies the following \emph{transversality condition}:
	\begin{align}	\label{transodts}
	\lim_{N \rightarrow \infty} \frac{1}{N}\mathbb{E}^{\Omega}\left[ V^\ast\left(Q\left(N \right),E \left(N \right) \right) |Q\left(0 \right),E \left(0 \right)\right]=0
\end{align}
	Then, we have the following results:
	\begin{itemize}
		\item $\theta^\ast= \underset{\Omega}{\min} 	\ \overline{D}\left(\Omega\right)$ is the optimal average cost for any initial state  $\boldsymbol{\chi}\left(0 \right) $ and $V^\ast\left(Q,E \right)$ is called the \emph{priority function}.
		\item  Suppose there exists an admissible stationary control policy $\Omega^*$ with $\Omega^*\left(\boldsymbol{\chi} \right) = p^\ast$ for any $\boldsymbol{\chi} $, where $p^\ast$ attains the minimum of the R.H.S. of (\ref{OrgBel}) for given $\boldsymbol{\chi}$. Then,  the optimal control policy of Problem \ref{IHAC_MDP} is given by $\Omega^*$.~\hfill\IEEEQED
	\end{itemize}	
\end{Theorem}

\begin{proof}
	please refer to Appendix A.
\end{proof}

Based on the unichain assumption of the control policy in Definition \ref{adddtdomain}, there is a unique solution for the Bellman equation in (\ref{OrgBel}) and the transversality condition in (\ref{transodts}). The solution $V^\ast \left(Q,E\right)$ captures the dynamic priority of the data flow for different  $(Q,E)$. However, obtaining the priority function $V^\ast \left(Q,E\right)$ is highly non-trivial as it involves solving a large system of nonlinear fixed point equations. Brute-force approaches (such as value iteration and policy iteration \cite{mdp2})) have huge complexity. 

\vspace{-0.1cm}
\framebox{\begin{minipage}[t]{0.95\columnwidth}
	
	 	\vspace{0.1cm}  \emph{Challenge 1: }Huge complexity in obtaining  the priority function $V^\ast \left(Q,E\right)$.\vspace{0.1cm} 
	
\end{minipage}}

\section{Virtual Continuous Time System and Approximate Priority  Function}
In this section, we  adopt a continuous time approach so that we can exploit calculus techniques and theories of differential equations   to obtain a closed-form approximate  priority function. Specifically, we first reverse-engineer a virtual continuous time system (VCTS) and an associated total cost problem in the VCTS. We show that the optimality conditions of the VCTS is equivalent to that of the original MDP (up to $o(\tau)$ order optimal). Based on that, we exploit calculus techniques and theories of differential equations  to obtain a closed-form characterization of the priority function $V^\ast(Q,E)$.

\subsection{Virtual Continuous Time System} \label{VCTS}
We first define the VCTS, which is a fictitious system with a continuous virtual queue state $\left(q(t),e(t)\right)$,  where  $q\left(t\right) \in [0,\infty)$ and $e\left(t\right) \in [0,N_E)$ are the virtual data queue state and virtual energy queue   state at time $t$ ($t\in[0,\infty)$).  

Let $\Omega^v$ be the virtual power control policy  of the VCTS. Similarly, $\Omega^v$ is $\mathcal{F}_t^v$-adapted, where  $\mathcal{F}_t^v=\sigma\big(\big\{h(s),q(s),e(s): 0<s<t\big\}\big)$ and $\big\{\mathcal{F}_t^v\big\}$ is the filtration of the VCTS.    Furthermore, the virtual power control policy $\Omega^v$ satisfies the virtual energy availability constraint, i.e., $p(t) \tau \leq e(t)$ ($\forall t$). Given an initial virtual system state $\left(q_0,e_0\right)$  and a virtual policy $\Omega^v$,  the  trajectory of the  virtual queueing system is  described by the following  coupled differential equations with reflections:
\begin{align}
		&\mathrm{d} q \left(t \right)  = \left(- \mathbb{E}\left[  R\left( h \left(t \right), p \left(t \right) \right) \big|q \left(t \right) , e \left(t \right) \right] +\overline{\lambda} \ \right) \tau\mathrm{d}t + \mathrm{d} L \left(t \right)	\label{queue1}	\\
		&\mathrm{d} e \left(t \right)  = \big(-\mathbb{E}\left[  p\left(t\right) \big|q \left(t \right) , e \left(t \right)\right] +\overline{\alpha} \ \big)\tau \mathrm{d}t - \mathrm{d} U \left(t \right)		\label{queue2}
	\end{align}
	where $L \left(t \right)$ and $U \left(t \right)$ are the \emph{reflection processes}\footnote{$L \left(t \right)$ and $U \left(t \right)$  are non-decreasing and minimal subject to the constraint that $q \left(t \right) \geq 0$ and $e \left(t \right) \leq N_E$, respectively \cite{bbmsf}.} associated with the lower data queue boundary $q (t)=0$ and upper energy queue boundary $e (t)=N_E$, which are uniquely determined by the following equations \txtblue{(\emph{Chap. 2.4} of \cite{bbmsf})}:
\begin{align}	
	& L \left(t \right) = \max \bigg\{0, - \min_{t' \leq t} \bigg[q_0	\label{lowref} \\
	&+ \int_0^{t'}\left(- \mathbb{E}\left[  R\left(h \left(s \right), p \left(s \right)\right) \big|q \left(s \right) , e \left(s \right)\right] +\overline{\lambda} \ \right) \tau \mathrm{d}s \bigg] \bigg\}  \notag \\
	& U \left(t \right) = \max \bigg\{0,  \max_{t' \leq t} \bigg[e_0 \label{uppref} \\
	&+ \int_0^{t'} \big(-\mathbb{E}\left[  p\left(s\right)\big|q \left(s \right) , e \left(s \right) \right] +\overline{\alpha} \ \big)\tau\mathrm{d}s \bigg]- N_E \bigg\}	\notag  	
\end{align}
with $L \left(0 \right)=U \left(0 \right)=0$.

\begin{figure}[t]
\centering
\subfigure[Trajectories of $\{q (t), L (t)\}$.]{
\includegraphics[width=2.5in,angle=90]{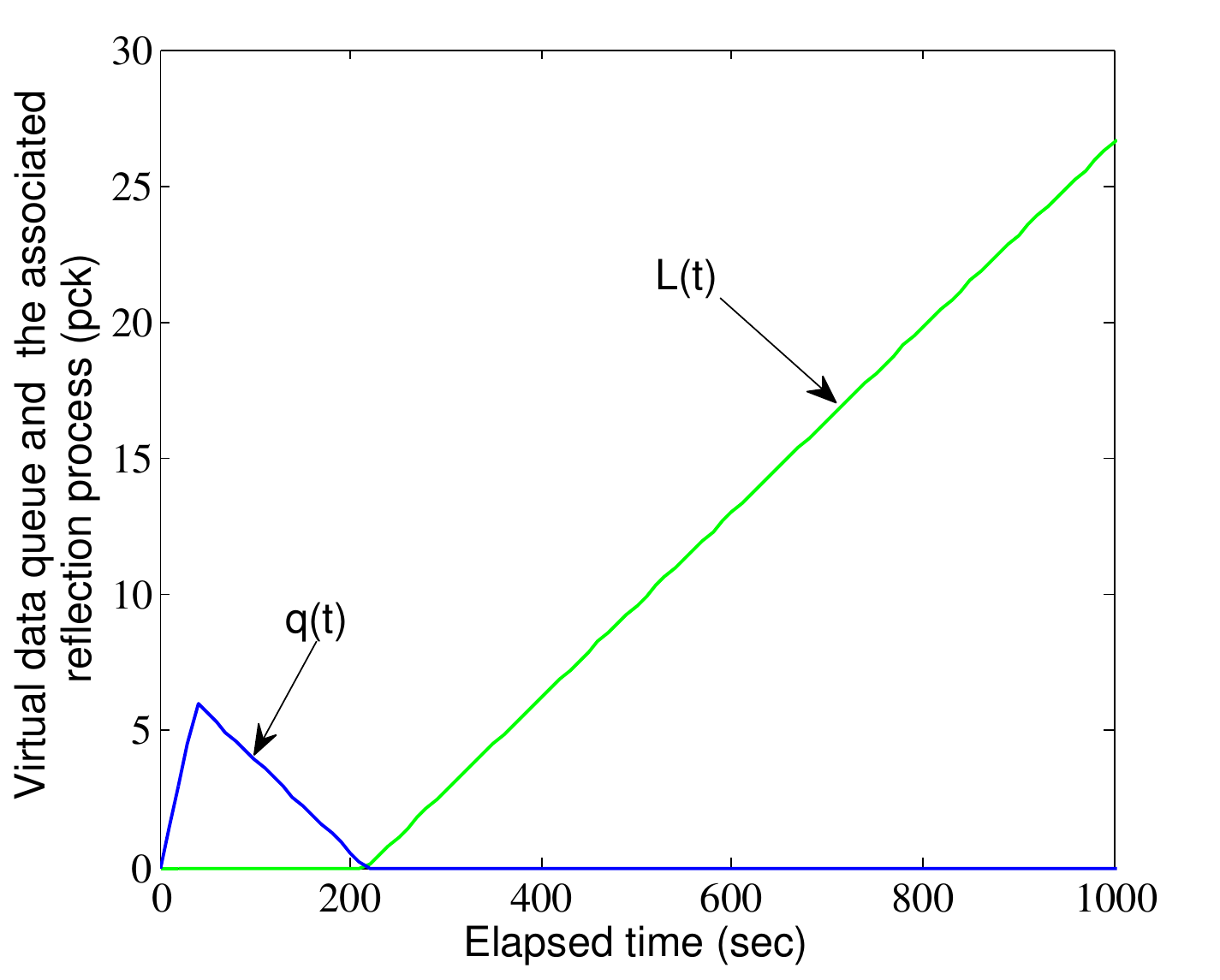}}
\hspace{-0.5cm}	\centering
\subfigure[Trajectories of $\{e (t), U (t)\}$.]{
\includegraphics[width=3.1in]{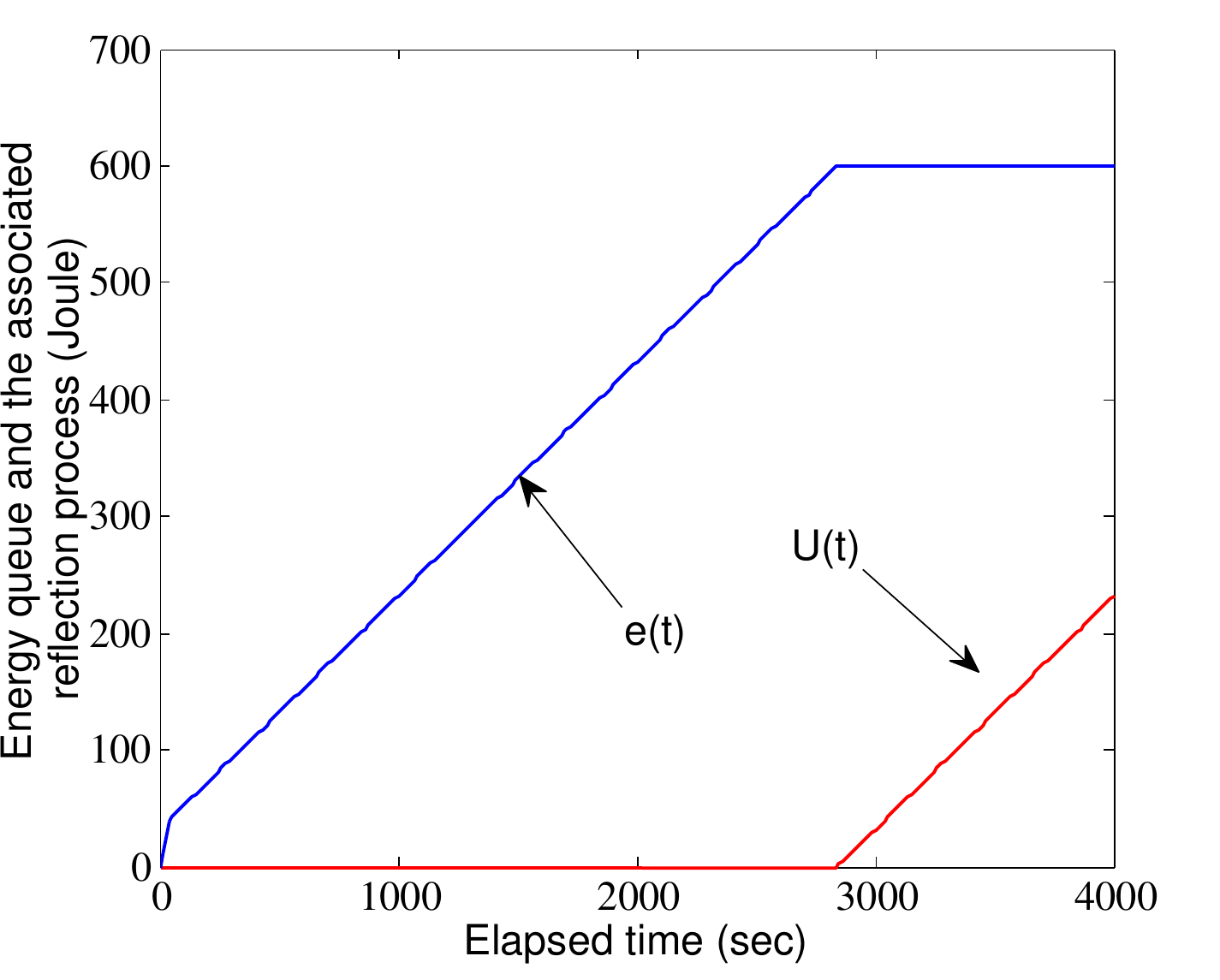}}
\caption{The system parameters are configured as follows: $\tau=0.1$ s, $\overline{\lambda}=1.5$ pcks/s, $\overline{\alpha}=10$ W, $N_E=600$ J. The virtual control policy ${\Omega^v}$ is $p=0$ W when $e<3.5$ J,  and $p=8$ W if $e>40$ J.}\vspace{-0.5cm}
\label{egillu}	
\end{figure}

Note that the process $L \left(t \right)$ ensures that the virtual  data queue length $q \left(t \right)$ will not go below zero. The process   $U \left(t \right)$ together with the virtual energy availability constraint  ensures that  the virtual energy queue length  lies in the domain $[0, N_E]$.  Fig.~\ref{egillu} illustrates\footnote{According to \cite{energymetric},  commercial solar panels usually provide 1$\sim$10 mW/cm$^2$ energy harvesting performance. We assume that  the wireless transmitter (e.g., base station)  is equipped with a 20cm$\times$50cm solar panel. Therefore, it has at most 10W energy harvesting capability.} an example of the trajectories of $\left\{q (t), L (t)\right\}$ and $\left\{e (t), U (t)\right\}$ for a virtual policy $\Omega^v$.

 Furthermore, we have the following definition on the admissible  virtual control policy for the VCTS. 
\begin{Definition}	[Admissible Virtual Control Policy for VCTS]		\label{addvctss}
	A virtual policy $\Omega^v$ for the VCTS is  admissible if  the following requirements are satisfied:
	\begin{itemize}
		\item  For any  initial virtual queue state $\left(q_0,e_0\right)$, the virtual queue trajectory $\big(q \left(t \right),e \left(t \right)\big)$ in (\ref{queue1}) and (\ref{queue2}) under $\Omega^v$ is unique.
		\item For any  initial virtual queue state $\left(q_0,e_0\right)$, the total cost $\txtblue{\int_0^{\infty} q \left(t \right)   \mathrm{d}t}$ under $\Omega^v$ is bounded.~\hfill\IEEEQED
	\end{itemize}
\end{Definition}

\subsection{Total Cost Problem under the VCTS}
Given an admissible  virtual control policy $\Omega^v$, we define the total cost of the VCTS starting from a given  initial virtual queue state $\left(q_0,e_0\right)$  as
\begin{equation}		\label{totalU}
	V \left(q_0,e_0;{\Omega^v}\right)  = \txtblue{\int_0^{\infty} q\left(t \right)   \mathrm{d}t}
\end{equation}

We consider the following  infinite horizon total cost problem for the VCTS:
\begin{Problem}		[Infinite Horizon Total Cost Problem for  VCTS]  \label{fluid problem1}
For any initial  virtual queue state  $\left(q_0,e_0\right)$, the infinite horizon total cost problem for the VCTS is formulated as
	\begin{align}
		\min_{\Omega^v} V\left(q_0,e_0;{\Omega^v}\right) 
	\end{align}
	where $V \left(q_0,e_0;{\Omega^v}\right) $ is given in (\ref{totalU}).~\hfill\IEEEQED
\end{Problem}

Note that the  two technical conditions in Definition \ref{addvctss}  on  the  admissible virtual policy  are for the existence of an optimal policy for the total cost problem in Problem \ref{fluid problem1}.  The above total cost problem has been well-studied in the continuous time optimal control theory \txtblue{(\emph{Chap. 2.6} of  \cite{KRbook})}. The solution can be obtained by solving the \emph{Hamilton-Jacobi-Bellman} (HJB) equation as  below. 
\begin{Theorem}	[Sufficient Conditions for Optimality under VCTS]	\label{HJB11}
	Assume there exists a function  $V\left( q, e\right)$ that is of class\footnote{$f(\mathbf{x})$ ($\mathbf{x}$ is a $K$-dim vector) is of class $\mathcal{C}^1(\mathbb{R}_+^K)$, if the  first order partial derivatives  w.r.t. each element of $\mathbf{x}\in \mathbb{R}_+^K$ are  continuous. } $\mathcal{C}^1(\mathbb{R}_+^2)$, and   $V\left( q, e\right)$ satisfies  the following HJB equation:
	\begin{align}	\label{cenHJB}
		&\min_{p \leq e/\tau} \ \mathbb{E}\left[\frac{q}{\overline{\lambda}\tau}+   \frac{\partial V\left( q, e\right)}{\partial q} \left(-   R\left( h , p\right) +\overline{\lambda} \ \right)  \right. \\
		&\left.\hspace{2cm} + \frac{\partial V\left( q, e\right)}{\partial e} \big(-p +\overline{\alpha} \ \big)  \bigg| q,e \right]=0 \quad \forall q, e \notag 
	\end{align}
Furthermore, for all admissible virtual control policy $\Omega^v$ and initial virtual queue state $\left(q_0,e_0\right)$, the following conditions are satisfied:
\txtblue{\begin{align}
\left\{
	\begin{aligned}	 \label{trankern}
		& \limsup_{T \rightarrow \infty } \int_0^T \frac{\partial V\left( 0, e\left(t \right)\right)}{\partial q} L\left(t \right) \mathrm{d}t  = 0	\\
		& \limsup_{T \rightarrow \infty } \int_0^T \frac{\partial V\left( q\left(t \right), N_E\right)}{\partial e} U\left(t \right) \mathrm{d}t  = 0	\\
		& \limsup_{T \rightarrow \infty } V\left(q\left(T \right), e\left(T \right) \right)= 0
	   \end{aligned}
   \right.
\end{align}}Then, we have the following results:
\begin{itemize}
	\item	$V\left(q,e\right) =\min_{\Omega^v}  V \left(q_0,e_0;{\Omega^v}\right) $ is the optimal total cost when $(q_0,e_0)=(q,e)$ and $V\left(q,e\right)$ is called the \emph{virtual priority function}.
	\item 	 Suppose there exists an admissible virtual stationary control policy   $\Omega^{v \ast}$ with $\Omega^{v \ast}(h,q,e)=p^\ast$ for any  $(h,q,e)$, where  $p^\ast$ attains the minimum of the L.H.S. of (\ref{cenHJB}) for given  $(h,q,e)$.   Then,  the optimal control policy of Problem \ref{fluid problem1} is given by $\Omega^{v \ast}$.\KO
\end{itemize}
\end{Theorem}
\begin{proof}
	Please refer to Appendix B.
\end{proof}

In the following theorem, we establish the relationship between the virtual priority  function $V(Q,E)$ in Theorem \ref{HJB11}  and the optimal  priority  function $V^\ast(Q,E)$ in Theorem \ref{LemBel}.
\begin{Theorem}	[Relationship  between  $V(Q,E)$ and  $V^\ast(Q,E)$]	\label{them111}
	If $V(Q,E)=\mathcal{O}\left(Q^2 \right)$ and $\Omega^{v \ast}$ is admissible in the discrete time system, then  $V^\ast\left(Q,E \right)=V\left(Q,E \right)+o(\tau)$.~\hfill\IEEEQED
\end{Theorem}
\begin{proof}  
please  refer to Appendix C.	
\end{proof}

Theorem \ref{them111} means that   $V(Q,E)$ can serve as an approximate priority function  to the optimal priority function $V^\ast\left(Q,E \right)$ with approximation error  $o(\tau)$.  As a result, solving the optimality equation in (\ref{OrgBel}) is transformed into a calculus problem of solving the HJB equation in (\ref{cenHJB}).  In the next subsection, we shall focus on solving the HJB equation in (\ref{cenHJB}) by leveraging the well-established theories of calculus and differential equations.

\section{Closed-Form  Delay-Optimal Power  Control}	\label{seclow}
The HJB equation in Theorem \ref{HJB11}  is a coupled two-dimensional partial differential equation (PDE) and hence, one key obstacle is to obtain the closed-form solution to the PDE.

\framebox{\begin{minipage}[t]{0.95\columnwidth}

	 	\vspace{0.1cm}  \emph{Challenge 2: }Solution of the coupled two-dimensional PDE in Theorem \ref{HJB11}.\vspace{0.1cm} 
\end{minipage}}

\vspace{0.2cm}

In this section, using asymptotic analysis, we obtain closed-form solutions to the multi-dimensional PDE in different operating regimes. We
also discuss the control insights from the structural properties of the closed-form priority functions for different asymptotic regimes.

\subsection{General Solution}
We first have the following corollary on the optimal power control  based on the HJB equation in Theorem \ref{HJB11} for given $V(q,e)$: 
\vspace{-0.8cm}
\begin{center}
\textcolor{blue}{}
\framebox{\begin{minipage}[t]{1\columnwidth}
\begin{Corollary}	[Optimal Power Control based on Theorem \ref{HJB11}]	\label{corooptk}
	For given priority function $V(q,e)$, the optimal power control action from the HJB equation in Theorem \ref{HJB11} is given by 
	\begin{align}	\label{optpow}
		p^\ast =\min\left\{ \left(-{\frac{\partial V\left( q, e\right)}{\partial q}}\bigg/{\frac{\partial V\left( q, e\right)}{\partial e}}-\frac{1}{|h|^2}\right)^+, \frac{e}{\tau}\right\}
\end{align}\KO
\end{Corollary}
\end{minipage}}
\par\end{center}

\begin{Remark}	[Structure of the Optimal Power Control Policy]
	The optimal  power control policy in (\ref{optpow}) depends on the instantaneous CSI, DQSI and EQSI. Furthermore, the power control action has a \emph{multi-level water-filling} structure as illustrated in Fig.~\ref{figillu1}--Fig.~\ref{figillu2}, where the water level is adaptive to the DQSI and the EQSI indirectly via the priority function $V\left( q, e\right)$. Therefore, the function $V(q,e)$ captures how the DQSI and the EQSI affect the overall priority of the data flow. \KO
\end{Remark}

We then  establish the following theorem on the sufficient conditions to ensure the existence of solution to the PDE in Theorem \ref{HJB11}:
\begin{Theorem}	[Sufficient Conditions for the Existence of Solution]	\label{nonempty}
	There exists a $V(q,e)$ that satisfies (\ref{cenHJB}) and (\ref{trankern})  in Theorem \ref{HJB11}  if the following conditions are satisfied:
	\begin{align}		
		& \overline{\lambda}  < \exp\left( \frac{1}{x}\right)E_1 \left(\frac{1}{x} \right)		\label{energyrequir}	\\	
		 & N_E  \geq e^\ast 	\label{energyrequir1}
	\end{align}
	where $x$ satisfies $x\exp\left(-\frac{1}{x} \right)  - E_1 \left(\frac{1}{x} \right) = \overline{\alpha}$ and $E_1(x) \triangleq \int_1^{\infty} \frac{e^{-tx}}{t}\mathrm{d}t$ is the exponential integral function.  $e^\ast$ is the solution of the fixed point equation in (\ref{ssd2asd}) in Appendix C if $\overline{\lambda}> E_1\left(\frac{1}{\overline{\alpha}} \right) $, and $e^\ast=\overline{\alpha}\tau$ if $\overline{\lambda} \leq E_1\left(\frac{1}{\overline{\alpha}} \right) $.\KO
\end{Theorem}
\begin{proof}	
	Please refer to Appendix D.
\end{proof}

The challenge  is to find a priority function $V(q,e)$ that satisfies (\ref{cenHJB}) and (\ref{trankern}).  Note that the PDE in (\ref{cenHJB}) is a two-dimensional PDE, which has no closed-form solution for the priority function $V(q,e)$. In the next subsection \ref{asysolfinn},  we consider different asymptotic regimes and obtain closed-form  solutions   of $V(q,e)$  for these operating   regimes.

\begin{figure}[t]
  \centering
  \includegraphics[width=3in]{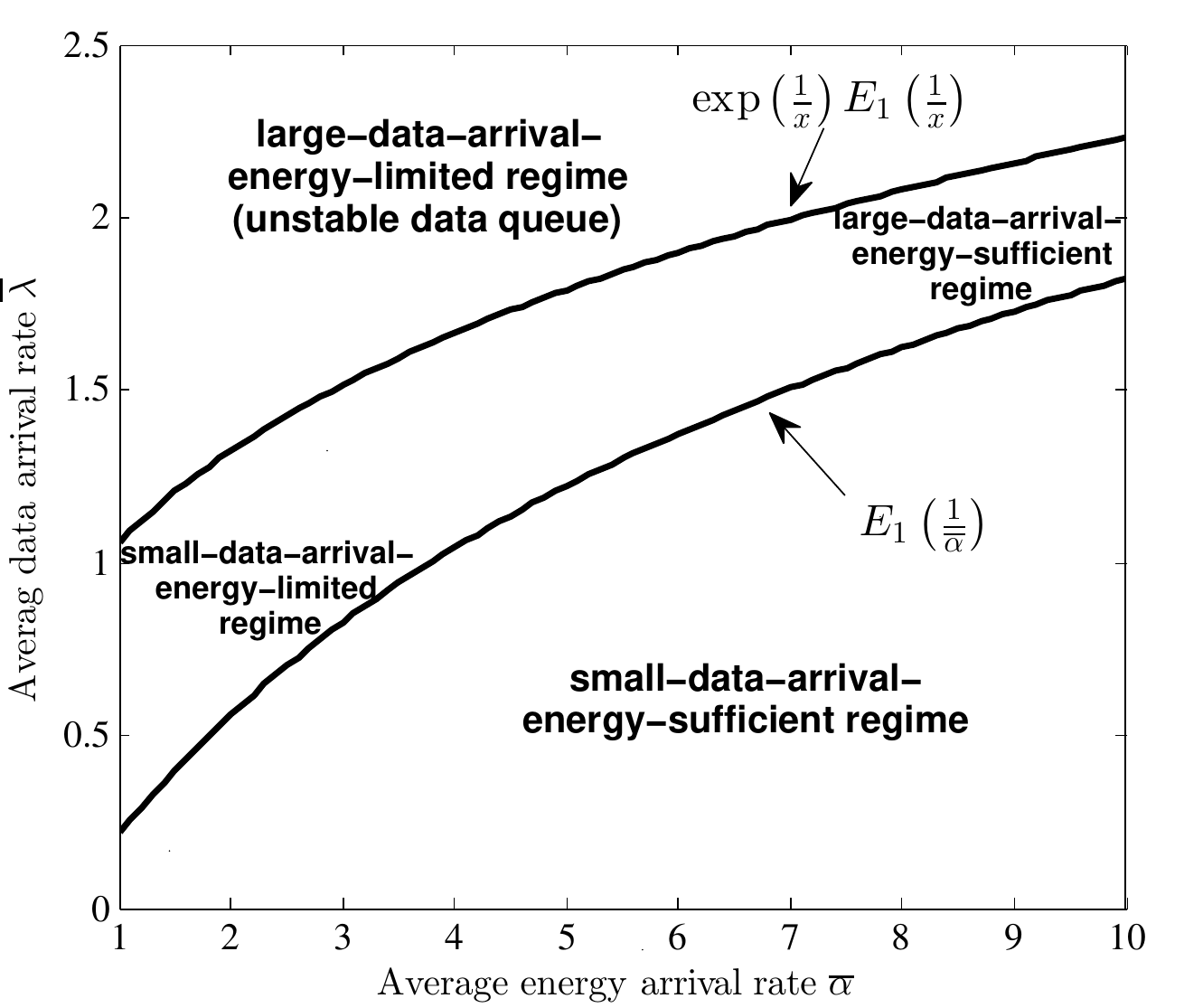}
  \caption{\tblue{Asymptotic regimes of the energy harvesting system.}}\vspace{-0.5cm}
  \label{regm}
\end{figure}

\subsection{Asymptotic Closed-Form Priority Functions and Control Insights}	\label{asysolfinn}
In this subsection, we obtain the closed-form priority functions $V(q,e)$  in different asymptotic regimes\footnote{Under the condition in (\ref{energyrequir}) in Theorem \ref{nonempty}, we have that $\overline{\alpha}$ grows at least at the order of $\exp(\overline{\lambda})$. Therefore, large   $\overline{\lambda}$ induces  large $\overline{\alpha}$. The regime  with large  $\overline{\lambda}$ and small  $\overline{\alpha}$  will cause the  system to be unstable  and  is not included in our discussions.} as illustrated in Fig.~\ref{regm} and discuss the control insights for each regime.

\subsubsection{\txtblue{Large-Data-Arrival}-Energy-Sufficient  Regime} 
In this regime, we consider the operating  region with  large  $\overline{\lambda}$ and large $\overline{\alpha}$, and $E_1\left(\frac{1}{\overline{\alpha}} \right)< \overline{\lambda} <\exp\left( \frac{1}{x}\right)E_1 \left(\frac{1}{x} \right)$ (where $x$ satisfies $x\exp\left(-\frac{1}{x} \right)  - E_1 \left(\frac{1}{x} \right) = \overline{\alpha}$). This regime corresponds to the scenario that we have a \txtblue{large data arrival rate}  for the data queue and sufficient renewable energy supply  for the energy queue to maintain the data queue stable. The closed-form priority function $V(q,e)$ for this regime is given by the following theorem:
\begin{Theorem} \emph{(Closed-Form $V(q,e)$  for the \txtblue{Large-Data-Arrival}-Energy-Sufficient Regime)}	\label{sym111}
	Under the \txtblue{large-data-arrival}-energy-sufficient  regime, the closed-form $V(q,e)$ of the PDE in Theorem \ref{HJB11} is given by
	\begin{itemize}
		\item When $0<e<e^{th}$ ($e^{th}$ is the solution of $E_1 \left( \frac{\tau}{e^{th}} \right)=\overline{\lambda}$), we have
			\begin{align}
				V(q,e)= \frac{e^2}{4 \overline{\lambda}\overline{\alpha}^2 \tau} \left(1+2\gamma_{eu}+2 \overline{\lambda} - 2 \log \frac{e}{\tau}\right)-\frac{eq}{\overline{\lambda} \overline{\alpha} \tau}+C_1
	\label{finalsol1}
			\end{align}
	where $\gamma_{eu}$ is the Euler's constant and $C_1=\frac{\tau}{4 \overline{\lambda}} \left(1+2\gamma_{eu}+2 \overline{\lambda} - 2 \log \overline{\alpha}\right)$.
		\item When $e\geq e^{th}$, $V(q,e)$ is a function of $q$ only.\KO
	\end{itemize}
\end{Theorem}

\begin{proof}
	Please refer to Appendix E.
\end{proof}

Based on Theorem \ref{sym111}, when $e\geq e^{th}$, since $V(q,e)$  is a function of $q$ only, we have $\frac{\partial V\left( q, e\right)}{\partial e}=0$ for  given $q$ and $e$. Therefore, the water level in (\ref{optpow}) is infinite and hence, we have $p^\ast=\frac{e}{\tau}$. Furthermore,  based on the closed-form priority function in (\ref{finalsol1}), we can calculate the closed-form expression of the water level\footnote{From (\ref{finalsol1}), we have $-{\frac{\partial V\left( q, e\right)}{\partial q}}\big/{\frac{\partial V\left( q, e\right)}{\partial e}}=\frac{\overline{\alpha}e}{e\left(\gamma_{eu}+\overline{\lambda} - \log(\frac{e}{\tau}) \right)-\overline{\alpha}q}$.} $-{\frac{\partial V\left( q, e\right)}{\partial q}}\big/{\frac{\partial V\left( q, e\right)}{\partial e}}$ in (\ref{optpow}).  We summarize the optimal power control structure for this regime in the following corollary:
\vspace{-1cm}
\begin{center}
\textcolor{blue}{}
\framebox{\begin{minipage}[t]{1\columnwidth}
\begin{Corollary}	\emph{(Optimal Power Control Structure  for the \txtblue{Large-Data-Arrival}-Energy-Sufficient  Regime)}	\label{cor1sd111}
	The optimal power control for  the \txtblue{large-data-arrival}-energy-sufficient  regime is given by 
	\begin{itemize}
		\item When $0< e < e^{th}$ and $q>\frac{e}{\overline{\alpha}}\left(\gamma_{eu}+\overline{\lambda} - \log(\frac{e}{\tau}) \right)$, $p^\ast = 0$.
		\item When $0< e < \overline{\alpha}  q$ and $q<\frac{e}{\overline{\alpha}}\left(\gamma_{eu}+\overline{\lambda} - \log(\frac{e}{\tau}) \right)$,  the water level $-{\frac{\partial V\left( q, e\right)}{\partial q}}\big/{\frac{\partial V\left( q, e\right)}{\partial e}}$ is an increasing function of  $q$ for a  given $e$, and is  a decreasing function of  $e$ for a given $q$.
		\item When $  \overline{\alpha}  q<e<e^{th}$ and $q<\frac{e}{\overline{\alpha}}\left(\gamma_{eu}+\overline{\lambda} - \log(\frac{e}{\tau}) \right)$,  the water level $-{\frac{\partial V\left( q, e\right)}{\partial q}}\big/{\frac{\partial V\left( q, e\right)}{\partial e}}$ is an increasing  function of  $q$ for a given $e$, and is  an increasing function of  $e$ for a  given $q$.		\item When $e\geq e^{th}$, $p^\ast=\frac{e}{\tau}$.\KO
	\end{itemize}
\end{Corollary}
\end{minipage}}
\par\end{center}

\begin{proof}
	Please refer to Appendix F.
\end{proof}

\begin{figure}
  \centering
  \includegraphics[width=3.5in]{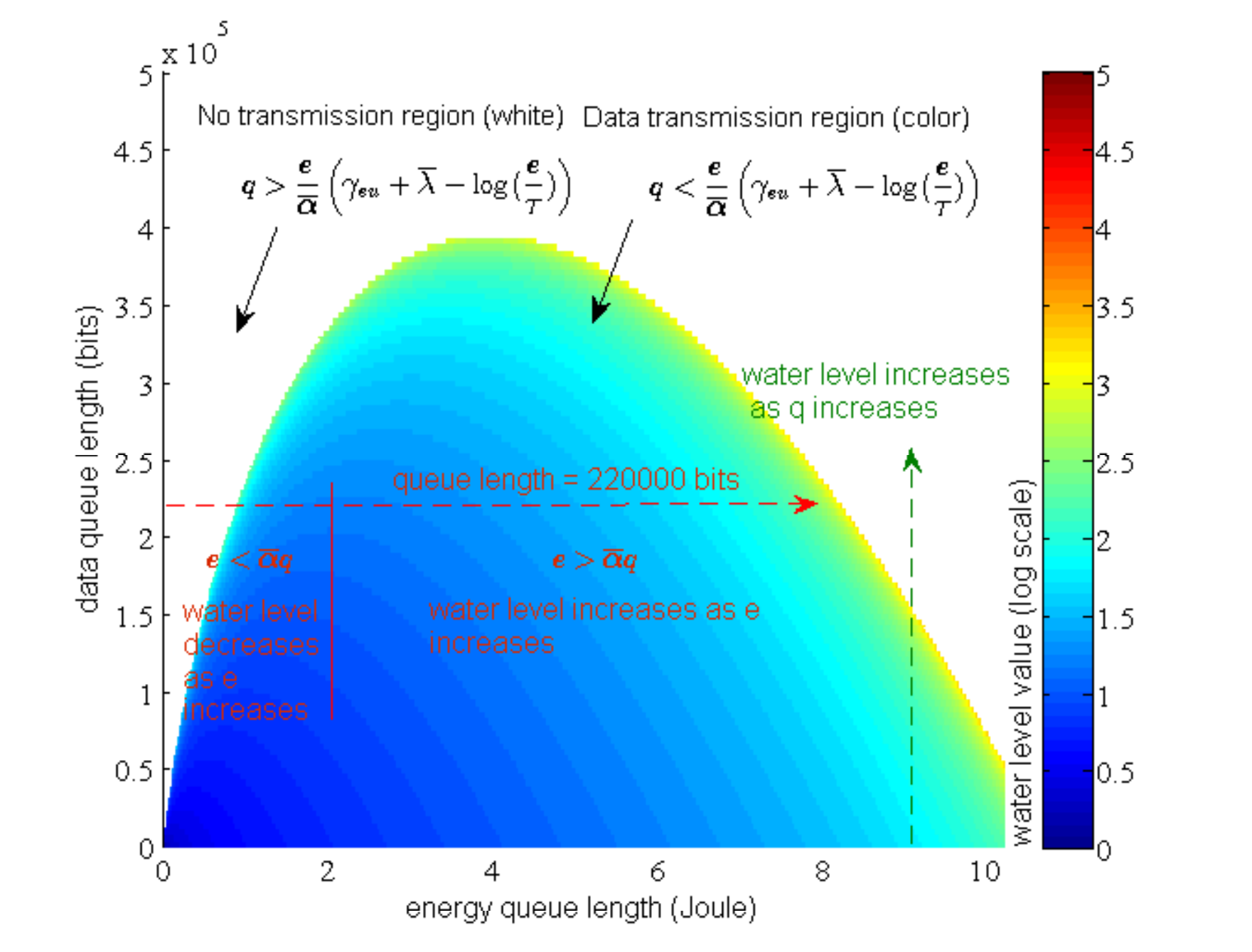}
  \caption{Water level versus the data queue length and the energy queue length for the \txtblue{large-data-arrival}-energy-sufficient regime, where $\tau = 0.1$ s, $\overline{\lambda}=1.8$ pcks/s, $\overline{\alpha}=10$ W, bandwidth is $1$ MHz, and average packet length  is $1$ Mbits.}\vspace{-0.5cm}
  \label{figillu1}\vspace{-0.5cm}
\end{figure}

Fig.~\ref{figillu1} illustrates the water level versus the data queue length and the energy queue length when $e<e^{th}$. Specifically, Corollary \ref{cor1sd111} means that when $0< e < e^{th}$ and for a large data  queue length  $q>\frac{e}{\overline{\alpha}}\left(\gamma_{eu}+\overline{\lambda} - \log(\frac{e}{\tau}) \right)$, we do not use any renewable energy  to transmit data. The reason is that we do not have enough energy to support the \txtblue{large data arrival rate}, and it is appropriate to wait for  future  good transmission opportunities.  For a small queue length $q<\frac{e}{\overline{\alpha}}\left(\gamma_{eu}+\overline{\lambda} - \log(\frac{e}{\tau}) \right)$,  we can use the available energy  for transmission and the  water level is increasing w.r.t. $q$, which is in accordance with the high urgency of the data flow. Furthermore, when\footnote{In order for $q<\frac{e}{\overline{\alpha}}\left(\gamma_{eu}+\overline{\lambda} - \log(\frac{e}{\tau}) \right)$ to hold, we require $\overline{\alpha}q \leq \tau \exp(\gamma_{eu}+\overline{\lambda}-1)$. For large $\overline{\lambda}$, $e^{th}\approx \tau \exp(\gamma_{eu}+\overline{\lambda})$. Therefore, we have $\overline{\alpha}q \in [0, e^{th}]$.} $0< e < \overline{\alpha}  q$, the water level decreases as $e$ increases, which is reasonable because it is better to save some energy for the future transmissions. When $e > \overline{\alpha}  q$,  the water level increases as $e$ increases, which is reasonable because we have relatively sufficient available energy  and it is appropriate to use more power to decrease the data queue. When $e\geq e^{th}$, we have sufficient renewable energy, and it is appropriate  to use all the available energy and make room for the future energy arrivals.

\subsubsection{\txtblue{Small-Data-Arrival}-Energy-Limited Regime} 
In this regime, we consider the operating  region with  small  $\overline{\lambda}$ and small $\overline{\alpha}$, and $E_1\left(\frac{1}{\overline{\alpha}} \right)< \overline{\lambda} <\exp\left({\frac{1}{x}} \right) E_1\left(\frac{1}{x} \right)$. This regime corresponds to the scenario that we have a \txtblue{small data arrival rate}  for the data queue and  insufficient  energy supply  for the energy queue. The closed-form priority function $V(q,e)$ for this regime is given by the following theorem:
\begin{Theorem} \emph{(Closed-Form $V(q,e)$  for the \txtblue{Small-Data-Arrival}-Energy-Limited Regime)}	\label{sym1}
	Under the \txtblue{small-data-arrival}-energy-limited regime, the closed-form $V(q,e)$ of the PDE  in Theorem \ref{HJB11}   is given by
	\begin{itemize}
		\item When $0<e<e^{th}$ ($e^{th}$ is the solution of $E_1 \left( \frac{\tau}{e^{th}} \right)=\overline{\lambda}$), we have
			\begin{align}
		V(q,e)=-\frac{ e^3}{3   \overline{\lambda}\overline{\alpha}^2 \tau^2}+ \frac{ e^2}{2 \overline{\alpha}^2 \tau} - \frac{q e}{\overline{\lambda}\overline{\alpha}\tau}	+C_2	\label{finalsol2}
	\end{align}
		\item When $e\geq e^{th}$, $V(q,e)$ is a function of $q$ only and $C_2=\frac{\tau}{2}-\frac{\overline{\alpha}\tau}{3\overline{\lambda}}$.\KO
	\end{itemize}
\end{Theorem}

\begin{proof}
	Please refer to Appendix G.
\end{proof}

Based on the closed-form $V(q,e)$ in Theorem \ref{sym1}, we have the following corollary summarizing the optimal power control structure for this regime\footnote{From (\ref{finalsol2}), we have $-{\frac{\partial V\left( q, e\right)}{\partial q}}\big/{\frac{\partial V\left( q, e\right)}{\partial e}}=\frac{\overline{\alpha}\tau e}{-e^2+\overline{\lambda}\tau e - \overline{\alpha}\tau q}$.}:
\vspace{-0.5cm}
\begin{center}
\textcolor{blue}{}
\framebox{\begin{minipage}[t]{1\columnwidth}
\begin{Corollary}	\emph{(Optimal Power Control Structure  for the \txtblue{Small-Data-Arrival}-Energy-Limited Regime)}	\label{cor1wewe111}
	The optimal power control for  the \txtblue{small-data-arrival}-energy-limited regime is given by 
	\begin{itemize}
		\item When $0< e < e^{th}$ and $q>\frac{-e^2 + \overline{\lambda} \tau e}{\overline{\alpha}\tau}$, $p^\ast = 0$.
		\item When $0< e < \sqrt{\overline{\alpha}\tau q}$ and $q<\frac{-e^2 + \overline{\lambda} \tau e}{\overline{\alpha}\tau}$,  the water level $-{\frac{\partial V\left( q, e\right)}{\partial q}}\big/{\frac{\partial V\left( q, e\right)}{\partial e}}$ is an increasing function of  $q$ for a  given $e$, and is  a decreasing function of  $e$ for  a given $q$.
		\item When $  \sqrt{\overline{\alpha}\tau q}<e<e^{th}$ and $q<\frac{-e^2 + \overline{\lambda} \tau e}{\overline{\alpha}\tau}$,  the water level $-{\frac{\partial V\left( q, e\right)}{\partial q}}\big/{\frac{\partial V\left( q, e\right)}{\partial e}}$ is an increasing  function of  $q$ for a given $e$, and is  an increasing function of  $e$ for a  given $q$.
		\item When $e\geq e^{th}$, we have $p^\ast=\frac{e}{\tau}$.	\KO
	\end{itemize}
\end{Corollary}
\end{minipage}}
\par\end{center}

\begin{proof}
	Please refer to Appendix H.
\end{proof}

\begin{figure}
  \centering
  \includegraphics[width=3.5in]{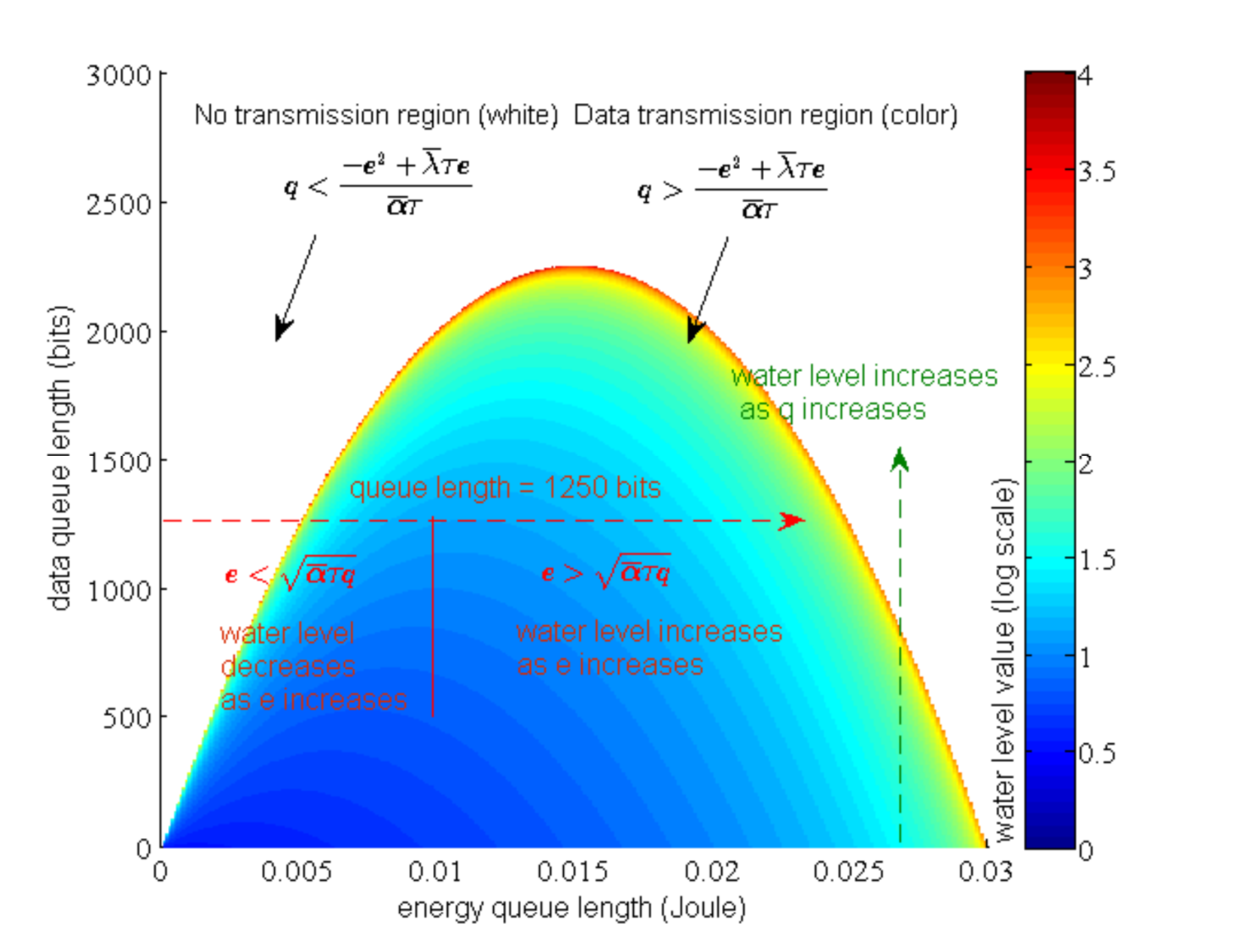}
  \caption{Water level versus the data queue length and the energy queue length for the \txtblue{small-data-arrival}-energy-limited regime, where $\tau = 0.1$ s, $\overline{\lambda}=0.3$ pcks/s, $\overline{\alpha}=1$ W, bandwidth is $1$ MHz, and average packet length  is $1$ Mbits.}\vspace{-0.5cm}
  \label{figillu2}
\end{figure}

Fig.~\ref{figillu2} illustrates the water level versus the data queue length and the energy queue length when $e<e^{th}$. Specifically,  Corollary \ref{cor1wewe111} means that when $0< e < e^{th}$ and for a large data  queue length  $q>\frac{-e^2 + \overline{\lambda} \tau e}{\overline{\alpha}\tau}$, we do not use any renewable energy  to transmit data. The reason is that even though we can use  the limited  energy for data  transmission, the data queue length will not decrease significantly, which contributes very little to the delay performance. Instead, if we do not use the energy at the current slot, we can save it and wait for the future good transmissions opportunities. On the other hand, for a small queue length $q<\frac{-e^2 + \overline{\lambda} \tau e}{\overline{\alpha}\tau}$, we can use the available energy  for transmission and the  water level is increasing w.r.t. $q$, which is in accordance with the high urgency of the data flow.  Furthermore, when\footnote{In order for $q<\frac{-e^2 + \overline{\lambda} \tau e}{\overline{\alpha}\tau}$ to hold, we require $\sqrt{\overline{\alpha}\tau q} \leq \frac{\overline{\lambda}\tau}{2}$. For small $\overline{\lambda}$, $\frac{e^{th}}{\tau} \exp\left(-\frac{\tau}{e^{th}} \right)\approx \overline{\lambda}<\frac{e^{th}}{\tau}\Rightarrow e^{th}>\overline{\lambda}\tau$. Therefore, we have $\sqrt{\overline{\alpha}\tau q} \in [0, e^{th}]$.}  $0< e < \sqrt{\overline{\alpha}\tau q}$, large $e$ leads to a lower water level. This is reasonable because it is appropriate that for small $e$, we  can save some energy in the current slot for better transmission opportunities in the future slots. When $\sqrt{\overline{\alpha}\tau q}<e<e^{th}$, large $e$ leads to a higher  water level  because we  have sufficient available energy  and it is appropriate to use more power to decrease the data queue. When $e\geq e^{th}$, we have plenty of renewable energy, and it is sufficient to use all the available energy to support the \txtblue{small data arrival rate}.

\subsubsection{\txtblue{Small-Data-Arrival}-Energy-Sufficient  Regime} 
In this regime, we consider the operating  region with  $\overline{\lambda} \leq E_1\left(\frac{1}{\overline{\alpha}} \right)$.  This regime corresponds to the scenario that we have a \txtblue{small data arrival rate} for the data queue and sufficient renewable energy supply  in the energy queue  to maintain the data queue stable.  The closed-form priority function $V(q,e)$ for this regime  is given by the following theorem:
\begin{Theorem} \emph{(Closed-Form $V(q,e)$  for the \txtblue{Small-Data-Arrival}-Energy-Sufficient   Regime)}	\label{sym33}
	Under the  \txtblue{small-data-arrival}-energy-sufficient regime, the closed-form $V(q,e)$ of the PDE in Theorem \ref{HJB11}   is given by
	\begin{itemize}
		\item When $0<e<e^{th}$ ($e^{th}$ is the solution of $E_1 \left( \frac{\tau}{e^{th}} \right)=\overline{\lambda}$), we have
			\begin{align}
		V(q,e)=\frac{1}{2 \overline{\alpha}^2 \tau} e^2 - \frac{e q}{\overline{\lambda}\overline{\alpha}\tau}-\frac{1}{2\lambda^2\tau}\left({q}-\frac{\overline{\lambda}}{\overline{\alpha}}e\right)^2	\label{finalsol3}
	\end{align}
		\item When $e\geq e^{th}$, $V(q,e)$ is a function of $q$ only.\KO
	\end{itemize}
\end{Theorem}

\begin{proof}
	Please refer to Appendix I.
\end{proof}

Based on the closed-form $V(q,e)$ in Theorem \ref{sym33}, we have the following corollary summarizing the optimal power control structure in this regime\footnote{Based on (\ref{finalsol3}), we have ${\frac{\partial V\left( q, e\right)}{\partial e}}=0$ for all $q, e$, which induces an infinite water level in (\ref{optpow}). Hence, we have  $p^\ast=\frac{e}{\tau}$ when $0<e<e^{th}$.}:
\vspace{-1cm}
\begin{center}
\textcolor{blue}{}
\framebox{\begin{minipage}[t]{1\columnwidth}
\begin{Corollary}	\emph{(Optimal Power Control Structure  for the  \txtblue{Small-Data-Arrival}-Energy-Sufficient Regime)}	\label{cor1ssssd111}
	The optimal power control for  the    \txtblue{small-data-arrival}-energy-sufficient  regime is given by 
	\begin{align}
		p^\ast = \frac{e}{\tau}	\label{ppolicygreedy}
	\end{align}~\hfill~\IEEEQED
\end{Corollary}
\end{minipage}}
\par\end{center}

Corollary \ref{cor1ssssd111} means that the optimal control policy for the   \txtblue{small-data-arrival}-energy-sufficient regime is to use all the available energy in the energy buffer.  This is reasonable because in this regime we have $\overline{\lambda}\leq E_1\left(\frac{1}{\overline{\alpha}} \right)$, which means that there is plenty of renewable energy  and it is sufficient to use all the available energy to support the data traffic. 

\txtblue{Based on the closed-form solutions for the asymptotic operating regions in Theorem \ref{sym111}--\ref{sym33}, we propose the following  solution for the PDE in Theorem \ref{HJB11} that covers all regimes w.r.t. $\big(\overline{\lambda},\overline{\alpha}\big)$: 
\begin{align}	\label{solprameter}
	V(q,e)\approx \left\{\begin{aligned}
		&\text{sol. in Thm 5},\\
		& \qquad \overline{\alpha}\geq \overline{\alpha}^{th}, E_1\left(\frac{1}{\overline{\alpha}} \right)< \overline{\lambda} <\exp\left( \frac{1}{x}\right)E_1 \left(\frac{1}{x} \right)	\\
		&\text{sol. in Thm 6}, \\
		& \qquad \overline{\alpha}< \overline{\alpha}^{th}, E_1\left(\frac{1}{\overline{\alpha}} \right)< \overline{\lambda} <\exp\left( \frac{1}{x}\right)E_1 \left(\frac{1}{x} \right)	\\
		&\text{sol. in Thm 7},\qquad	\overline{\lambda}\leq E_1\left(\frac{1}{\overline{\alpha}} \right)
	\end{aligned}
	\right.
\end{align}
where  $\overline{\alpha}^{th}>0$ is a solution parameter.}

\subsection{Stability Conditions of using the Closed-Form Solution in the Discrete-Time System} 
 In the previous subsection, we  obtain the closed-form optimal power control solutions for different asymptotic regimes as in Theorem \ref{sym111}--\ref{sym33}. We then establish the following theorem on the stability conditions when using the control policy in Corollary \ref{corooptk} in the original discrete time system in (\ref{dataQ}) and (\ref{energyQ}):
\begin{Theorem}	\emph{(Stability Conditions of using the Closed-Form Solutions in the Discrete-Time System):}	\label{optimawlcontrolreg11}	Using (\ref{solprameter}) and the closed-form control policy in  Corollary \ref{corooptk}, if the following conditions are \txtblue{satisfied}:
	\begin{align}
		&\overline{\lambda}<\mathbb{E}\left[\exp\left(\frac{1}{\alpha} \right)E_1\left(\frac{1}{\alpha} \right)\right]	\label{energyrequqweir}\\
		&N_E\geq Ne^\ast	\label{energyrequqweirsad}
	\end{align}
	where $e^\ast$ is defined in (\ref{energyrequir1}), then the data queue in the original discrete time system in (\ref{dataQ}) is stable, in the sense that $\lim_{n\rightarrow \infty} \mathbb{E} \big[  Q^2(n)  \big] < \infty$. ~\hfill~\IEEEQED
\end{Theorem}

\begin{proof}
	Please refer to Appendix J.
\end{proof}

Theorem \ref{optimawlcontrolreg11} means that using (\ref{solprameter}), the  closed-form control policy in  Corollary \ref{corooptk} is admissible according to Definition \ref{adddtdomain}.
\begin{Remark}	[Interpretation of the Conditions in Theorem \ref{nonempty}]	\
	\begin{itemize}
		\item \textbf{Interpretation of the Condition on $\overline{\lambda}$ and $\overline{\alpha}$ in (\ref{energyrequqweir}):}	The condition in (\ref{energyrequqweir}) implies\footnote{$(\ref{energyrequqweir})\overset{(a)}{\Rightarrow}\overline{\lambda} <\exp\left({\frac{1}{\overline{\alpha}}} \right) E_1\left(\frac{1}{\overline{\alpha}} \right)=\mathcal{O}\left(\log\overline{\alpha}\right)$, where $(a)$ is due to $\mathbb{E}\left[\exp\left({\frac{1}{{\alpha}}} \right) E_1\left(\frac{1}{{\alpha}} \right)\right]<\exp\left({\frac{1}{\overline{\alpha}}} \right) E_1\left(\frac{1}{\overline{\alpha}} \right)$ using the concavity of  $\exp\left({\frac{1}{{x}}} \right) E_1\left(\frac{1}{{x}} \right)$ and the \emph{Jensen's Inequality}.  Therefore, $\overline{\alpha}$ grows at least at the order of $\exp(\overline{\lambda})$.}  that $\overline{\alpha}$ grows at least at the order of $\exp(\overline{\lambda})$. It indicates that for given $\overline{\lambda}$, if $\overline{\alpha}$  is too small, even if we use all the available energy in the energy buffer at each time slot, the average data arrival rate will be larger than the average data departure rate for the data queue buffer. Therefore, the data queue cannot be stabilized.  		
		\item \textbf{Interpretation of the Condition on $N_E$  in (\ref{energyrequqweirsad}):}	The condition in (\ref{energyrequqweirsad}) gives a first order design guideline on the dimensioning of the energy storage capacity  required at the transmitter. For example, $N_E$ should be at least at a similar order\footnote{From (\ref{ssd2}) in Appendix D, we have $e^\ast>\overline{\alpha}\tau$. Therefore, from (\ref{energyrequqweirsad}), we have $N_E>N \overline{\alpha}\tau $ which means that $N_E$ grows at least at the order of $N \overline{\alpha}\tau$.} of $N \overline{\alpha}\tau$. This condition on $N_E$  ensures that the  energy storage at the transmitter has sufficient energy to support data transmission for $N$ slots when $\alpha(t)$ is small.\KO
	\end{itemize}
\end{Remark}

\section{Simulations}
In this section, we compare the performance of the proposed closed-form delay-optimal power control scheme in (\ref{optpow}) with the following three baselines using numerical simulations:

\begin{itemize}
	\item \textbf{Baseline 1, Greedy Strategy (GS)\footnote{Baseline 1 (Baseline 2)  refers to the greedy policy (CSI dependent policy) in Section III (Section V) of \cite{queuestable}.} \cite{queuestable}:} At each time slot, the transmitter sends data to the receiver using the power $p(t)=\min \Big\{\overline{\alpha}- \epsilon, \frac{E(t)}{ \tau} \Big\}$ for a given small positive constant $\epsilon$.  The GS is a throughput-optimal policy in the stability sense, i.e., it  ensures the stability of the queueing network.	
	\item \textbf{Baseline 2, CSI-Only Water-Filling Strategy (COWFS) \cite{queuestable}:} At each time slot, the transmitter sends data to the receiver using the power $p(t)= \Big\{ \big(\frac{1}{\gamma}- \frac{1}{|h(t)|^2}\big)^+,\frac{E(t)}{ \tau} \Big\}$. Specifically, the water-filling solution in the COWFS is obtained by maximizing the ergodic capacity $\mathbb{E}\big[\log ( 1+p |h|^2 )  \big]$ with the average power constraint\footnote{The Lagrangian multiplier $\gamma$ for Baseline 2 and Baseline 3 can be obtained by the following iterative equation: $\gamma(t+1)=\left[\gamma(t)+ a_t \left(p - \overline{\alpha}+ \epsilon\right)\right]^+$, where $a_t$ is the  step size satisfying $\sum_t a_t=\infty$, $\sum_t a_t^2<\infty$. As $t\rightarrow \infty$, the convergent $\gamma(\infty)$ can be shown to satisfy  the average power constraint $\mathbb{E}[p]=\overline{\alpha}- \epsilon$ \cite{delaysurvey}.} $\mathbb{E}[p]=\overline{\alpha}- \epsilon$ for a  given small positive constant $\epsilon$.
	\item \textbf{Baseline 3, Queue-Weighted Water-Filling Strategy  (QWWFS) \cite{stab1}:} At each time slot, the transmitter sends data to the receiver using the power $p= \Big\{ \big(\frac{Q(t)}{\gamma}- \frac{1}{|h(t)|^2}\big)^+,\frac{E(t)}{ \tau} \Big\}$. The QWWFS is  also a  throughput-optimal policy. $\gamma$ is the  Lagrangian multiplier associated with the  average power constraint  $\mathbb{E}[p]=\overline{\alpha}- \epsilon$ for a  given small positive constant $\epsilon$.
\end{itemize}

In the simulation, we consider a point-to-point energy harvesting system, where a base station  (BS) communicates with a mobile station. The BS is equipped with a 40cm$\times$50cm solar panel with energy harvesting performance\footnote{If the surrounding environment of the BS has sufficient sunlight, the energy harvesting performance is high. Otherwise, the energy harvesting performance is low \cite{energymetric}.} 1$\sim$10 mW/cm$^2$. We assume Poisson packet arrival with average packet arrival rate $\overline{\lambda}$ (pck/s) and an exponentially distributed random packet size with mean  $1$ Mbits. The decision slot duration $\tau$ is $50$ ms, and the total bandwidth is $1$ MHz. Furthermore,  we consider Poisson energy arrival \cite{queuestable} with average energy arrival rate $\overline{\alpha}=$ 1$\sim$10 W. We assume that  the block length of the energy arrival process is $N=6000$, i.e., the energy arrival  rate $\alpha(t)$ at the BS changes every 5 min and the renewable energy is stored in a 1.2V 2000 mAh lithium-ion battery. We compare the delay performance of the proposed scheme with the above three baselines.

\subsection{\txtblue{Choice of the Solution Parameter $\overline{\alpha}^{th}$ in (\ref{solprameter})}}

\txtblue{Fig.~\ref{compareees2} illustrates the performance loss ratio\footnote{The performance loss ratio is defined as $\frac{\text{Perf. of the proposed scheme}-\text{Perf. of the VIA}}{\text{Perf. of the VIA}}$.} versus the average energy arrival rate with the average data arrival  rate $\overline{\lambda}=\frac{1}{2}\left[E_1\left(\frac{1}{\overline{\alpha}} \right) + \exp\left( \frac{1}{x}\right)E_1 \left(\frac{1}{x} \right)\right]$.  It can be observed that using the solution in Theorem \ref{sym111}, the performance loss is small for large $\overline{\alpha}$ and it increases as $\overline{\alpha}$ decreases. In addition, using the solution in Theorem \ref{sym1}, the performance loss is small for small $\overline{\alpha}$ and it increases as $\overline{\alpha}$ increases. It can be observed that choosing $\overline{\alpha}^{th}\approx 3.6$ can keep the performance loss down to $6\%$ over the entire operating regime w.r.t. $(\overline{\lambda},\overline{\alpha})$. }

\begin{figure}
  \centering
  \includegraphics[width=2.8in]{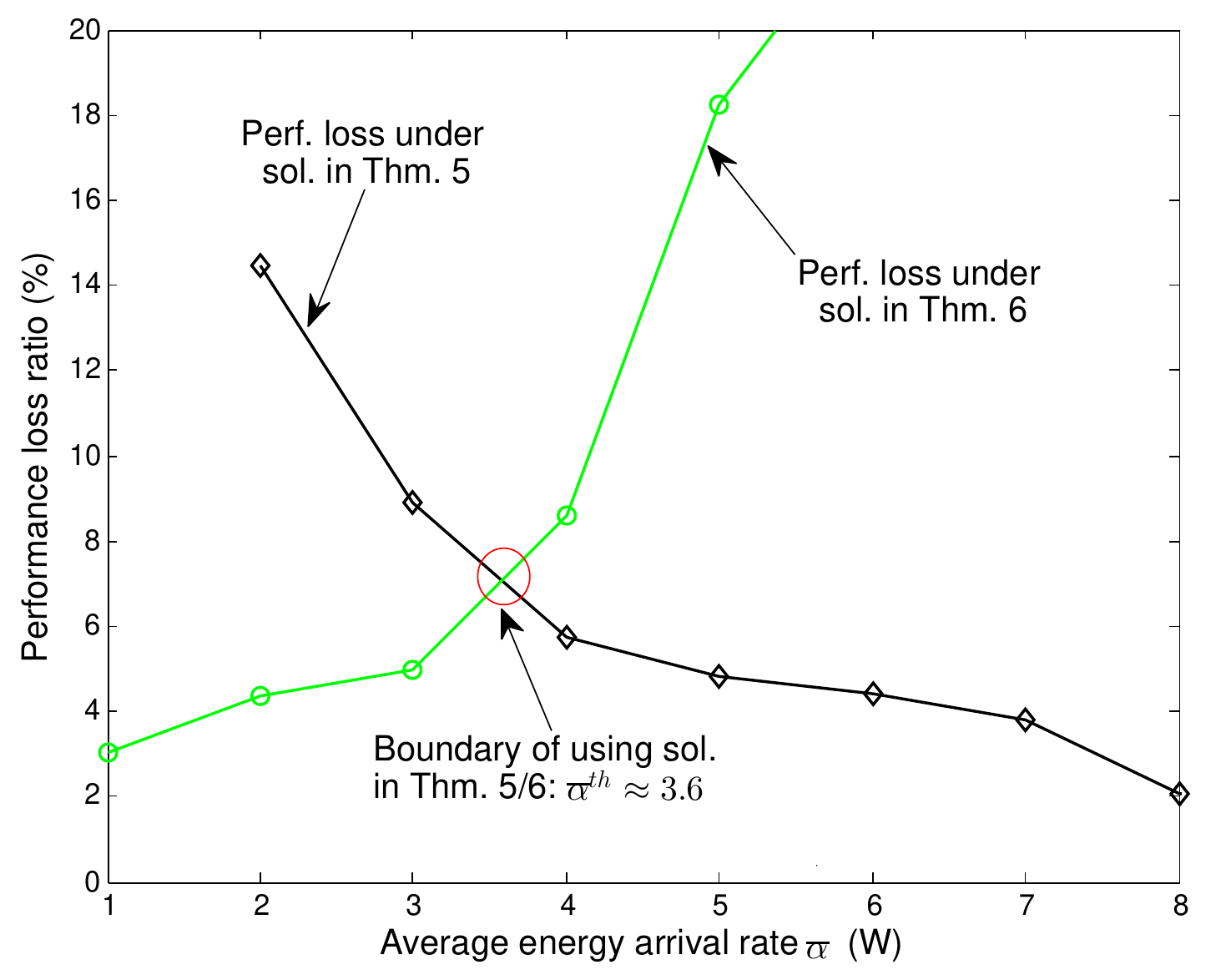}
  \caption{\txtblue{Performance loss ratio versus average energy arrival rate.}}
  \label{compareees2} \vspace{-0.5cm}
\end{figure}
\begin{figure}
  \centering
  \includegraphics[width=3in]{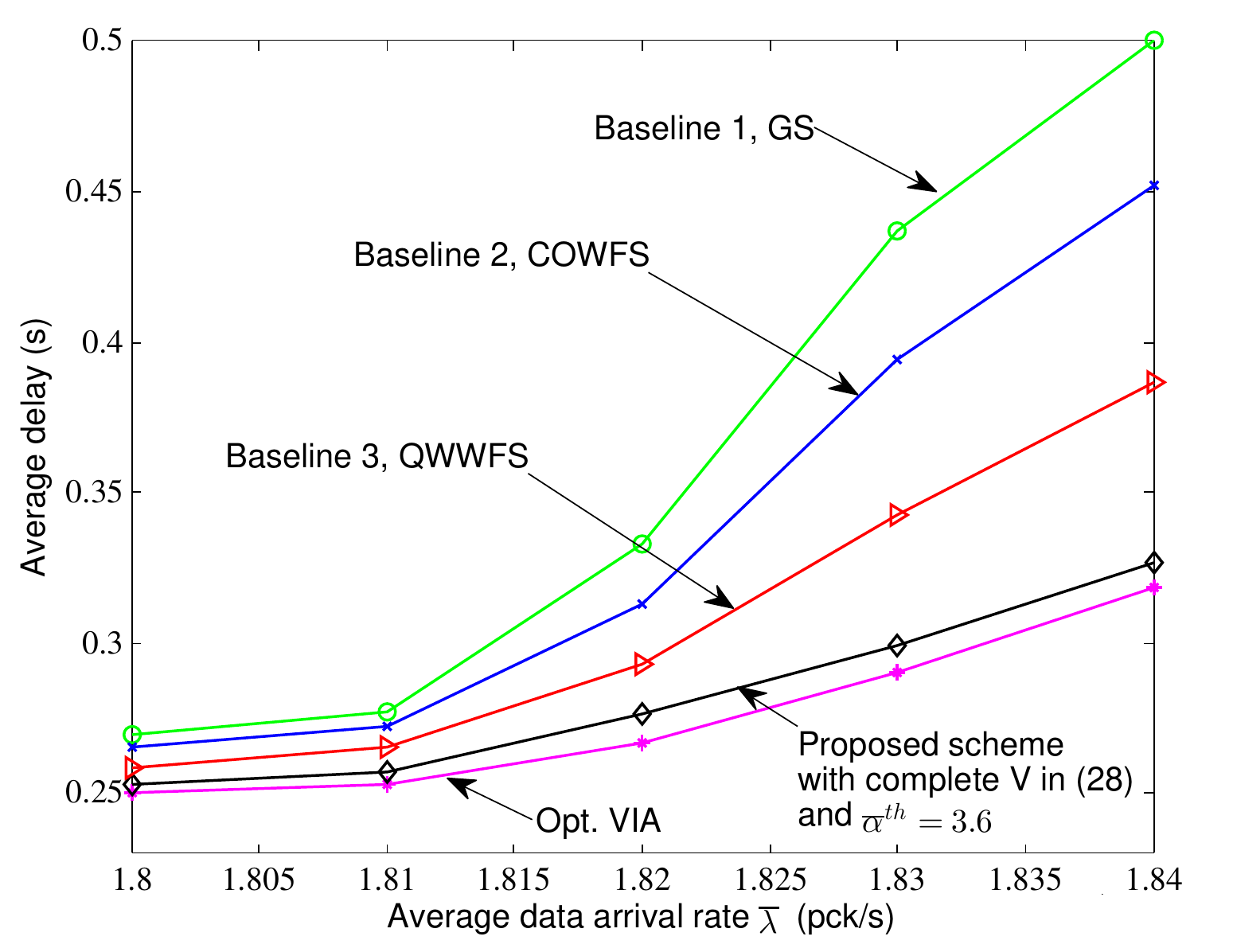}
  \caption{\txtblue{Average delay versus average data arrival rate  for the large-data-arrival-energy-sufficient regime. The average energy arrival is $\overline{\alpha}=10$ W.}}\vspace{-0.5cm}
  \label{fig13} 
\end{figure}

\subsection{Delay Performance for the \txtblue{Large-Data-Arrival}-Energy-Sufficient Regime}

Fig.~\ref{fig13} illustrates the average delay versus the average data arrival rate  for the \txtblue{large-data-arrival}-energy-sufficient regime. The average data arrival rate is $\overline{\lambda}=1.8\sim 1.84$ pcks/s and the average energy arrival rate is $\overline{\alpha}=10$ W. The average delay of all the schemes increases as the average data arrival rate increases, and the proposed scheme achieves significant performance gain over all the baselines. The gain is contributed by the DQSI and the EQSI aware dynamic water level structure.  \txtblue{It can be also observed that the performance of the proposed closed-form solution is very close to that of the optimal value iteration algorithm (VIA) \cite{mdp2}.}

\begin{figure}
  \centering
  \includegraphics[width=3in]{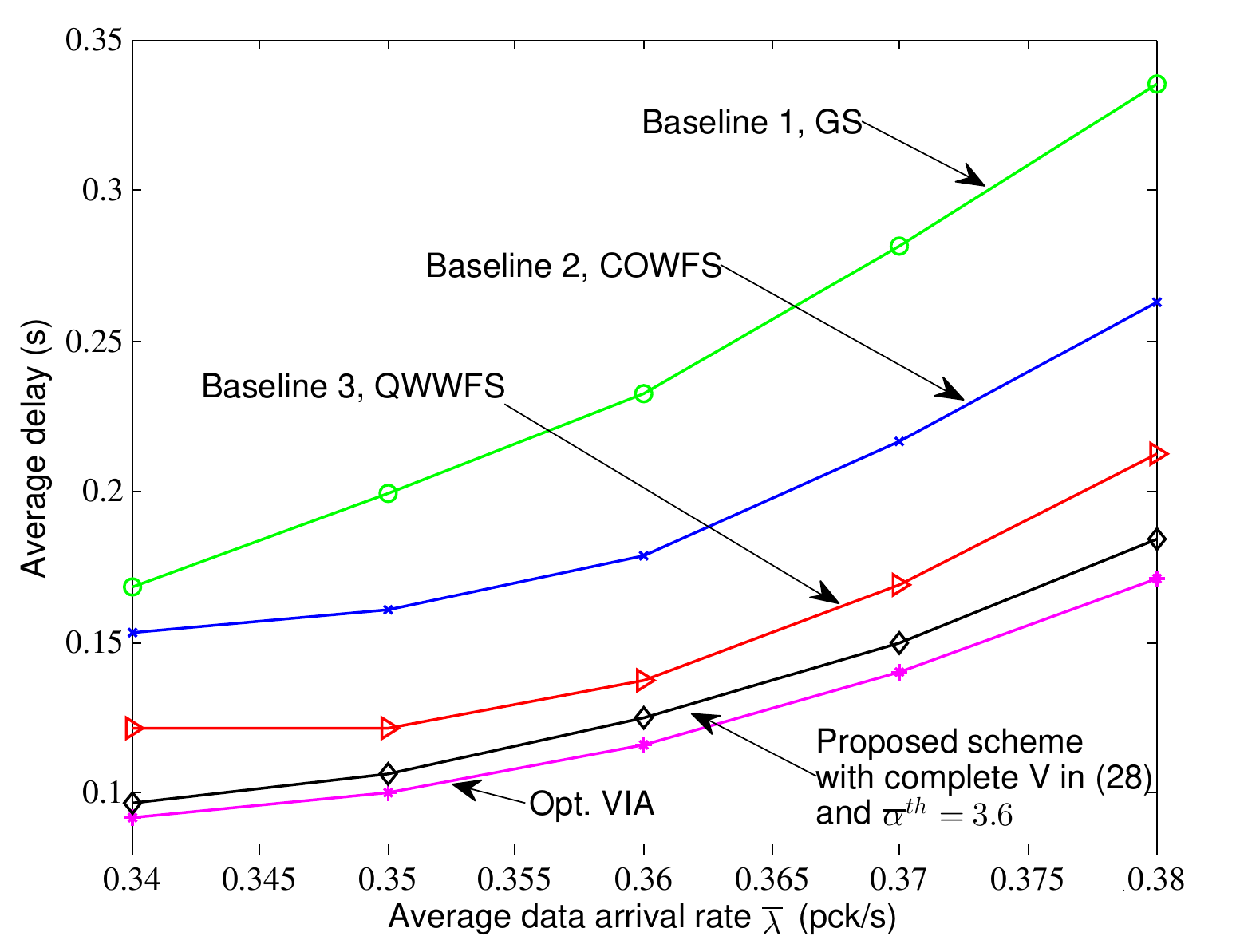}
  \caption{\txtblue{Average delay versus average data arrival rate  for the small-data-arrival-energy-sufficient regime. The average energy arrival is $\overline{\alpha}=1$ W.}}\vspace{-0.5cm}
  \label{fig23}
\end{figure}

\subsection{Delay Performance for the \txtblue{Small-Data-Arrival}-Energy-Limited Regime}

Fig.~\ref{fig23} illustrates the average delay versus the average data arrival rate  for the \txtblue{small-data-arrival}-energy-limited regime. The average data arrival rate is $\overline{\lambda}=0.34\sim0.38$ pcks/s and the average energy arrival rate is $\overline{\alpha}=1$ W. The proposed scheme achieves significant performance gain over all the baselines due to the  DQSI and the EQSI aware dynamic water level structure. \txtblue{Furthermore, the performance of the proposed closed-form solution is very close to that of the VIA.}

\begin{figure}
  \centering
  \includegraphics[width=3in]{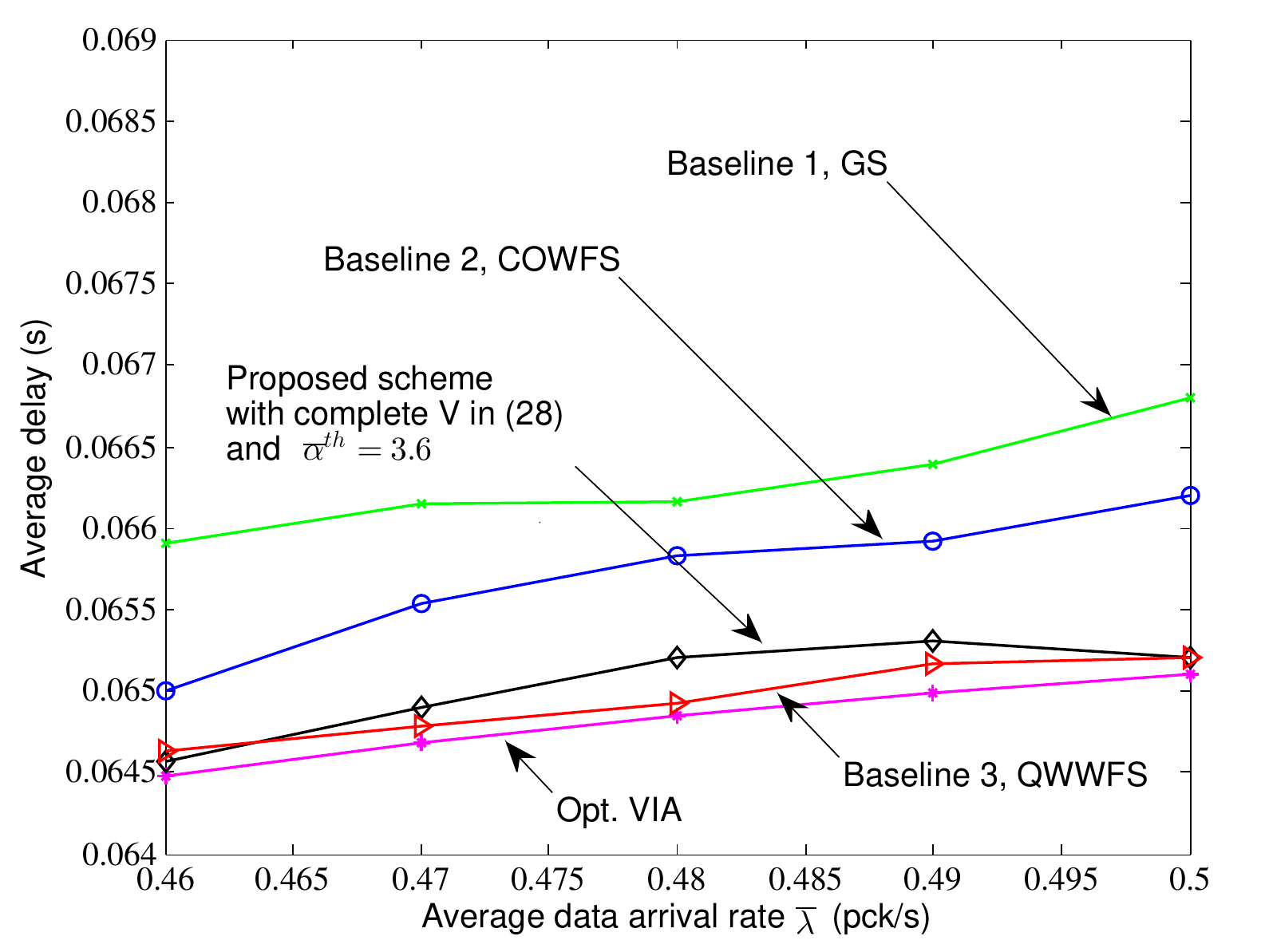}
  \caption{\txtblue{Average delay versus average data arrival rate  for the small-data-arrival-energy-limited regime. The average energy arrival is $\overline{\alpha}=6$ W.}}\vspace{-0.5cm}
  \label{fig33}
\end{figure}

\begin{table*}[t]
\centering
\begin{tabular}{|c|c|c|c|c|c|} 
\hline
\multicolumn{1}{|c}{} & \multicolumn{1}{|c}{Baseline 1} & \multicolumn{1}{|c}{Baseline 2} & \multicolumn{1}{|c}{Baseline 3} &  \multicolumn{1}{|c}{Proposed Scheme} &  \multicolumn{1}{|c|}{VIA} \\ \hline
Computational time ($N_E=2000$) & & & && 759s\\
Computational time ($N_E=4000$)& 0.2374ms & 1.729s& 15.437s & 0.2491ms& $>10^4$s \\
Computational time ($N_E=6000$)& & & &&  $>10^4$s \\	\hline
\end{tabular}
\caption{ {Comparison of the MATLAB computational time of the proposed scheme, the baselines and  the value iteration algorithm (VIA). The system parameters are configured as in Fig.~\ref{fig33}.}}\vspace{-0.5cm}
\label{tabletime}
\end{table*}

\subsection{Delay Performance for the \txtblue{Small-Data-Arrival}-Energy-Sufficient Regime}

Fig.~\ref{fig33} illustrates the average delay versus the average data arrival rate  for the \txtblue{small-data-arrival}-energy-sufficient regime. The average data arrival rate is $\overline{\lambda}=0.34\sim0.38$ pcks/s and the average energy arrival rate is $\overline{\alpha}=6$ W. The delay performance of the proposed scheme is very \txtblue{close} to that of  Baseline 3 and also better than  those of  Baselines 1 and 2.  However, our proposed scheme has lower complexity compared with Baseline 3, which  involves the gradient update to obtain  the Lagrangian multiplier. Therefore, it is better to adopt our proposed scheme for the \txtblue{small-data-arrival}-energy-sufficient regime. \txtblue{Furthermore,  the performance of the proposed closed-form solution is very close to that of the VIA.}

\subsection{Comparison of Complexity in Computational Time}

Table \ref{tabletime}  illustrates the comparison of the MATLAB computational time of the proposed solution, the  baselines and the brute-force VIA \cite{mdp2}. Note that the proposed scheme has similar complexity   to Baseline 1 due to the closed-form priority function. Therefore, our proposed scheme achieves significant performance gain with negligible computational cost.

\section{Summary}
In this paper, we propose a closed-form  delay-optimal power control solution for an energy harvesting wireless network with finite energy storage.  We formulate the associated stochastic optimization problem as an infinite horizon average cost MDP. Using a continuous time approach, we derive closed-form approximate priority functions  for different asymptotic regimes. Based on the closed-form approximations, we propose a closed-form  optimal control policy, which has a multi-level water filling structure and the water level is adaptive to the DQSI and the EQSI. Numerical results show that the proposed power control scheme has much better  performance than the   baselines.

\section*{Appendix A: Proof of Theorem \ref{LemBel}}
Following \emph{Proposition 4.6.1} of \cite{mdp2},  the sufficient conditions for the optimality of Problem \ref{IHAC_MDP} is that there exists a ($\theta^\ast, \{ V^\ast\left(\boldsymbol{\chi}  \right) \}$) that satisfies the following Bellman equation and $V^\ast$ satisfies the transversality condition in (\ref{transodts})  for all  admissible control policy $\Omega$ and initial  state $\boldsymbol{\chi}  \left(0 \right)$:
\begin{align}
	&\theta^\ast + V^\ast\left(\boldsymbol{\chi} \right) = \min_{p<E/\tau} \Big[ \frac{Q}{\overline{\lambda}}+  \sum_{\boldsymbol{\chi}' } \Pr\big[ \boldsymbol{\chi}'\big| \boldsymbol{\chi}, p\big]  V^\ast\left(\boldsymbol{\chi}'\right)    \Big]	\label{equveryimportant} \\
	&= \min_{p<E/\tau} \Big[  \frac{Q}{\overline{\lambda}}+ \sum_{Q',E'} \sum_{h' }  \Pr \big[ Q',E'\big| \boldsymbol{\chi}, p \big] \Pr \big[h'\big]  V^\ast\left(\boldsymbol{\chi}'\right)    \Big]	\notag	
\end{align}
Then, $\theta^\ast= \underset{\Omega}{\min} 	\ \overline{D}\left(\Omega\right)$ is the optimal average cost for any initial state  $\boldsymbol{\chi}\left(0 \right) $. Furthermore,  suppose there exists an stationary admissible $\Omega^*$ with $\Omega^*\left(\boldsymbol{\chi} \right) = p^\ast$ for any $\boldsymbol{\chi} $, where $p^\ast$ attains the minimum of the R.H.S. in (\ref{equveryimportant}) for given $\boldsymbol{\chi}$. Then,  the optimal control policy of Problem \ref{IHAC_MDP} is given by $\Omega^*$.

Taking expectation w.r.t. $h$ on both sizes of (\ref{equveryimportant}) and denoting $V^\ast\left(Q,E  \right) = \mathbb{E}\big[V^\ast\left(\boldsymbol{\chi} \right) \big| Q,E \big]$, we obtain the equivalent Bellman equation in (\ref{OrgBel}) in Theorem \ref{LemBel}.

\section*{Appendix B: Proof of Theorem \ref{HJB11}}
 Suppose $V\left(q,e\right)$ is of class $\mathcal{C}^1(\mathbb{R}_+^2)$, we have  $\mathrm{d} V\left( q, e\right) = \frac{\partial  V\left( q,e\right)}{\partial q}\mathrm{d} q+\frac{\partial  V\left( q,e\right)}{\partial e}\mathrm{d} e$. Substituting  the dynamics in (\ref{queue1}) and (\ref{queue2}), we obtain 
\begin{align}
	& \mathrm{d} V\left( q(t), e(t)\right) = D^{\Omega^v} \left(  V\left( q(t), e(t)\right) \right)  \mathrm{d}t   \\
	&\hspace{1cm} +\frac{\partial  V\left( q(t),e(t)\right)}{\partial q} \mathrm{d} L\left(t \right) -\frac{\partial  V\left( q(t),e(t)\right)}{\partial e} \mathrm{d} U\left(t \right)\notag
\end{align}	
where  $D^{\Omega^v}\left(  V\left( q, e\right) \right) \triangleq \frac{\partial  V\left( q,e\right)}{\partial q}\left(- \mathbb{E}\left[  R\left( h , p \right) \big|q,e\right]  +\overline{\lambda}    \right)\tau+\frac{\partial  V\left( q,e\right)}{\partial e}\big(-\mathbb{E}\left[  p\big|q,e\right] +\overline{\alpha}  \big)\tau$. Integrating  on both sizes w.r.t. $t$ from 0 to $T$, we have
\begin{align}
	& V\left( q(T),e(T)\right)-V\left( q_0,e_0\right)	\label{43ers} \\
	=&\int_0^T D^{\Omega^v}\left(  V\left( q(t), e(t)\right) \right) \mathrm{d}t+\int_0^T \frac{\partial  V\left( q(t),e(t)\right)}{\partial q} \mathrm{d} L\left(t \right)\notag \\
	&-\int_0^T \frac{\partial  V\left( q(t),e(t)\right)}{\partial e} \mathrm{d} U\left(t \right) \notag \\
	\overset{(a)}{=}&\int_0^T D^{\Omega^v}\left(  V\left( q(t), e(t)\right) \right) \mathrm{d}t+\int_0^T \frac{\partial  V\left(0,e(t)\right)}{\partial q} \mathrm{d} L\left(t \right)\notag \\
	&-\int_0^T \frac{\partial  V\left( q(t),N_E\right)}{\partial e} \mathrm{d} U\left(t \right) \notag \\
	=&\int_0^T \left(\frac{q(t)}{\overline{\lambda}}+D^{\Omega^v}\left(  V\left( q(t), e(t)\right) \right) \right)\mathrm{d}t+\int_0^T \frac{\partial  V\left(0,e(t)\right)}{\partial q} \mathrm{d} L\left(t \right)\notag \\
	&-\int_0^T \frac{\partial  V\left( q(t),N_E\right)}{\partial e} \mathrm{d} U\left(t \right) -\int_0^T \frac{q(t)}{\overline{\lambda}}\mathrm{d}t		\label{34impeuqq}
\end{align}
where $(a)$ is because $L(t)$ and $U(t)$ increase only when $q=0$ and $e=N_E$ according to \emph{Chapter 2.4} of \cite{bbmsf}. If $V(q,e)$  satisfies  (\ref{cenHJB}),   from (\ref{34impeuqq}), we have  for any admissible virtual policy $\Omega^v$,
\begin{align}
	&V\left( q_0,e_0\right)\leq   V\left( q(T),e(T)\right)-\int_0^T \frac{\partial  V\left(0,e(t)\right)}{\partial q} \mathrm{d} L\left(t \right)\notag \\
	&+\int_0^T \frac{\partial  V\left( q(t),N_E\right)}{\partial e} \mathrm{d} U\left(t \right) +\int_0^T \frac{q(t)}{\overline{\lambda}}\mathrm{d}t	\label{finalequapp2}
\end{align}

From the boundary conditions in (\ref{trankern}), we have $\limsup_{T \rightarrow \infty } \int_0^T \frac{\partial V\left( 0, e\left(t \right)\right)}{\partial q} \mathrm{d}L\left(t \right)   = 0$, $\limsup_{T \rightarrow \infty } \int_0^T\\  \frac{\partial V\left( q\left(t \right), N_E\right)}{\partial e} \mathrm{d} U\left(t \right) = 0$ and $\limsup_{T \rightarrow \infty } V\left(q\left(T \right), e\left(T \right) \right)= 0$. Hence, taking  the limit superior as $T\rightarrow \infty$   in (\ref{finalequapp2}), we have 
\begin{align}
	&V\left( q_0,e_0\right)\leq \limsup_{T \rightarrow \infty } \int_0^T \frac{q(t)}{\overline{\lambda}}\mathrm{d}t
\end{align}
where the above equality is achieved if the admissible virtual stationary policy $\Omega^{v}(q,e,h)$ attains the minimum in the HJB equation in (\ref{cenHJB}) for all $(q,e,h)$. Hence, such $\Omega^{v}$ is the optimal control policy of the total cost problem in VCTS in Problem \ref{fluid problem1}.

\section*{Appendix C: Proof of Theorem  \ref{them111}}

\vspace{0.1cm}
\hspace{-0.3cm} \emph{A. Relationship between the Discrete Time and VCTS  Optimality Equations}
\vspace{0.2cm}

We first prove the following corollary based on Theorem \ref{LemBel}.
\begin{Corollary}	[Approximate Optimality Equation]	\label{cor1}
	Suppose there exist $J\left( Q,E \right)$ of class $\mathcal{C}^1(\mathbb{R}_+^K)$ that solve the following \emph{approximate optimality equation}:
\begin{align}\label{bellman2}
	&\min_{p \leq E/\tau} \ \mathbb{E}\left[\frac{Q}{\overline{\lambda}\tau}+   \frac{\partial J\left( Q, E\right)}{\partial Q} \left(-   R\left( h , p\right) +\overline{\lambda} \ \right) \right. \notag \\
	&\left. \hspace{1cm}+ \frac{\partial J\left( Q, E\right)}{\partial E} \big(-p +\overline{\alpha} \ \big)  \bigg| Q,E \right]=0 ,\quad \forall Q,E 
\end{align}
Furthermore, for all admissible  control policy $\Omega$ and initial queue state $Q(0),E \left(0 \right)$, the transversality condition  in (\ref{transodts}) is satisfied for $J$. Then, we have $ V^\ast\left(Q,E  \right)=J\left(Q,E  \right)+o(\tau)$.~\hfill\IEEEQED
\end{Corollary}
\begin{proof}	[Proof of Corollary \ref{cor1}]
We will establish the following Lemmas \ref{applema}--\ref{tenlemma} to prove Corollary \ref{cor1}. For convenience, denote 
\begin{align}
	&T_{\boldsymbol{\chi}}(\theta, J, p)=    \frac{Q}{\overline{\lambda}}  +  \sum_{Q',E'}\Pr \big[ Q',E '\big| \boldsymbol{\chi}, p\big]J\left(Q',E'\right)   \notag \\
	&\hspace{5cm}- J \left(Q,E  \right)-  \theta	\\
	&T_{\boldsymbol{\chi}}^\dagger(\theta, J, p)=\frac{Q}{\overline{\lambda}}+   \frac{\partial J\left( Q, E\right)}{\partial Q} \left(-   R\left( h , p\right)\tau +\overline{\lambda}  \right)\tau   \notag \\
	&\hspace{3cm}+ \frac{\partial J\left( Q, E\right)}{\partial E} \big(-p +\overline{\alpha}  \big) \tau-  \theta
\end{align}\vspace{-0.4cm}

\hspace{-0.5cm} \emph{Step 1, Relationship between $ T_{\boldsymbol{\chi}}(\theta, J, p)$ and $T_{\boldsymbol{\chi}}^\dagger(\theta, J, p)$:}
\vspace{0.2cm}

\begin{Lemma}	\label{applema}
	For any $\boldsymbol{\chi}$, $ T_{\boldsymbol{\chi}}(\theta, J, p)=T_{\boldsymbol{\chi}}^\dagger(\theta, J, p)+\nu  G_{\boldsymbol{\chi}}(J,p)$ for some smooth function $G_{\boldsymbol{\chi}}$ and $\nu=o(\tau)$.
\end{Lemma}
\begin{proof}	[Proof of Lemma \ref{applema}] Let $\left(Q (n+1),E(n+1)\right)=(Q', E')$ and $\left(Q (n),E(n)\right)=(Q, E)$. For  sufficiently small $\tau$, according to the dynamics in (\ref{dataQ}) and (\ref{energyQ}), we have the following Taylor expansion on $J\left( Q',E'\right)$ in (\ref{OrgBel}): 
\begin{align}
	&\mathbb{E}\left[ J\left( Q',E'\right) \big| Q,E  \right] 	=J\left( Q,E\right)+ \mathbb{E}\left[ \frac{\partial J\left( Q, E\right)}{\partial Q} \left(-   R\left( h , p\right) \right.\right. \notag \\
	&\left.\left.\hspace{1cm} +\overline{\lambda} \right) + \frac{\partial J\left( Q,E\right)}{\partial E} \big(-p +\overline{\alpha}  \big)  \bigg| Q,E \right]\tau+ o(\tau)	\label{37weq}
\end{align}
Substituting (\ref{37weq}) into $T_{\boldsymbol{\chi}}(\theta, J, p)$, we obtain  $ T_{\boldsymbol{\chi}}(\theta, J, p)=T_{\boldsymbol{\chi}}^\dagger(\theta, J, p)+\nu  G_{\boldsymbol{\chi}}(J,p)$ for some smooth function $G_{\boldsymbol{\chi}}$ and $\nu=o(\tau)$.
\end{proof}\vspace{0.3cm}

\hspace{-0.5cm} \emph{Step 2, Growth Rate of $\mathbb{E}\left[T_{\boldsymbol{\chi}}  (0, J)\big|Q,E\right]$:}
\vspace{0.2cm}

Denote 
\begin{align}
	T_{\boldsymbol{\chi}}(\theta, J)=\min_{  p} T_{\boldsymbol{\chi}}(\theta, J, p), \qquad T_{\boldsymbol{\chi}}^\dagger(\theta, J)=\min_{  p} T_{\boldsymbol{\chi}}^\dagger(\theta, J, p)	\label{zerofuncsa}		
\end{align}
Suppose $(\theta^\ast,V^\ast)$ satisfies the Bellman equation in (\ref{OrgBel}) and $(0,J)$ satisfies  (\ref{bellman2}), we have for any $\boldsymbol{\chi}$,
\begin{align}
	\mathbb{E}\big[T_{\boldsymbol{\chi}}(\theta^\ast, V^\ast)\big|Q,E\big]=0, \quad \mathbb{E}\big[T_{\boldsymbol{\chi}}^\dagger(0, J)\big|Q,E\big]=0	\label{zerofunc}
\end{align}
Then, we establish the following lemma.
\begin{Lemma}	\label{applemma}
$\mathbb{E}\left[T_{\boldsymbol{\chi}}  (0, J)\big|Q,E\right]=o(\tau)$,  $\forall {Q, E}$.
\end{Lemma}

\begin{proof}	[Proof of Lemma \ref{applemma}]
For any $\boldsymbol{\chi}$, we have $T_{\boldsymbol{\chi}} (0, J)=\min_{ p}\left[ T_{\boldsymbol{\chi}}^\dagger(0, J, p)+\nu  G_{\boldsymbol{\chi}}(J,p) \right] \geq \min_{ p} T_{\boldsymbol{\chi}}^\dagger(0, J, p) \\ + \nu \min_{ p} G_{\boldsymbol{\chi}}(J,p)$. On the other hand, $T_{\boldsymbol{\chi}} (0, J) \leq \min_{ p} T_{\boldsymbol{\chi}}^\dagger(0, J, p) + \nu G_{\boldsymbol{\chi}}(J,p^\dagger)$, where $p^\dagger= \arg \min_{p} T_{\boldsymbol{\chi}}^\dagger(0, J, p) $.

From  (\ref{zerofunc}), $\mathbb{E}\left[\min_{ p} T_{\boldsymbol{\chi}}^\dagger(0, J, p)\big|Q,E\right]=\mathbb{E}\left[T_{\boldsymbol{\chi}}^\dagger(0, J)\big|Q,E\right]=0$. Since $T_{\boldsymbol{\chi}}^\dagger(0, J, p) $ and $G_{\boldsymbol{\chi}}(J,p^\dagger) $  are all smooth and bounded functions, we have $\mathbb{E}\left[T_{\boldsymbol{\chi}}  (0, J)\big|Q,E\right] = \mathcal{O}(\nu)=o(\tau)$ for any $Q, E$.
\end{proof}

\vspace{0.3cm}
\hspace{-0.5cm} \emph{Step 3, Difference between $V^\ast(Q,E)$ and $J(Q,E)$:}
\vspace{0.2cm}

\begin{Lemma}		\label{tenlemma}
	Suppose $\mathbb{E}[T_{\boldsymbol{\chi}}(\theta^\ast, V^\ast)|Q,E] = 0$ for all $Q,E$ together with the transversality condition in (\ref{transodts})  has a unique solution $(\theta^*, V^\ast)$. If $J$ satisfies  (\ref{bellman2}) and  the transversality condition in (\ref{transodts}),  then $V^\ast(Q,E)-J(Q,E)=o(\tau)$.
\end{Lemma}
\begin{proof}	[Proof of Lemma \ref{tenlemma}]
	Suppose for some $(Q',E')$, we have $J\left(Q',E' \right)=V^\ast\left(Q',E' \right)+\alpha$ for some $\alpha \neq 0$ as $\tau \rightarrow 0$. Now let $\tau \rightarrow 0$. From Lemma \ref{applemma}, we have $\mathbb{E}\left[T_{\boldsymbol{\chi}}  (0, J)\big|Q,E\right]= 0$ for all $Q,E$ and also $J$ satisfies  the transversality condition in (\ref{transodts}). However, $J\left(Q',E' \right) \neq V^\ast\left(Q',E' \right)$ because of the assumption that $J\left(Q',E' \right)=V^\ast\left(Q',E' \right)+\alpha$. This contradicts  the condition that $(\theta^*, V^\ast)$ is a unique solution of  $\mathbb{E}[T_{\boldsymbol{\chi}}(\theta^\ast, V^\ast)|Q,E] = 0$ for all $Q,E$  and the transversality condition in (\ref{transodts}). Hence, we must have $V^\ast(Q,E)-J(Q,E)=o(\tau)$ for all  $Q,E $.\end{proof}
\end{proof}

\hspace{-0.3cm} \emph{B. Relationship between the Discrete Time  Optimality Equation and the HJB Equation}
\vspace{0.2cm}

First, if $V(Q,E)$ that is of class  $\mathcal{C}^1(\mathbb{R}_+^2)$ satisfies the optimality conditions of the total cost problem in VCTS (as shown in Theorem \ref{HJB11}), then it also satisfies (\ref{bellman2}) in Corollary \ref{cor1}. Second,  since $V\left(Q,E \right)=\mathcal{O}(Q^2)$, we have $\lim_{n \rightarrow \infty}\mathbb{E}^{\Omega}\left[V\left(Q(n),E (n)\right) \right]< \infty$ for any admissible policy $\Omega$ of the discrete time system according to Definition \ref{adddtdomain}.  Hence, $V\left(Q,E  \right)$ satisfies the transversality condition in (\ref{transodts}). Using Corollary \ref{cor1}, we have $V^\ast\left(Q,E  \right)=V\left(Q,E \right)+o(\tau)$.

\section*{Appendix D: Proof of Theorem  \ref{nonempty}}
First, we simplify the PDE in (\ref{cenHJB}).  The optimal control policy that minimizes the L.H.S. of (\ref{cenHJB}) is  $p^\ast =\min\big\{ \big(-{\frac{\partial V\left( q, e\right)}{\partial q}}\big/{\frac{\partial V\left( q, e\right)}{\partial e}}-\frac{1}{|h|^2}\big)^+, \frac{e}{\tau}\big\}$. Substituting it to the PDE in (\ref{cenHJB}), we have
\begin{align}		
	\mathbb{E}\left[\frac{q}{\overline{\lambda}\tau}+   \frac{\partial V\left( q, e\right)}{\partial q} \left(-   R\left( h , p^\ast\right) +\overline{\lambda} \ \right)\right. \notag \\
	\left. + \frac{\partial V\left( q, e\right)}{\partial e} \big(-p^\ast +\overline{\alpha} \ \big)  \bigg| q,e \right]=0\label{pdewe}
\end{align}
For  convenience, denote $V_q \triangleq \frac{\partial V\big( q, e\big)}{\partial q}$ and $V_e \triangleq \frac{\partial V\big( q, e\big)}{\partial e}$. We then calculate the expectations in (\ref{pdewe}): $\mathbb{E}\big[p^\ast\big] = \Big[\int_{\frac{-V_e}{V_q}}^{\frac{- \tau V_e }{  V_e e+ V_q \tau}} \big(\frac{V_q}{-V_e}-\frac{1}{x}\big) \exp\big({-x}\big) \mathrm{d}x + \int_{{\frac{- \tau V_e }{  V_e e+ V_q \tau}}}^\infty \frac{e}{\tau} \exp(x) \mathrm{d}x \Big]\mathbf{1}\big(\frac{V_q}{-V_e}>\frac{e}{\tau} \big)+ \Big[\int_{\frac{-V_e}{V_q}}^\infty \big(\frac{V_q}{-V_e}-\frac{1}{x}\big)    \exp\big({-x}\big) \mathrm{d}x  \Big]\mathbf{1}  \big(\frac{V_q}{-V_e}<\frac{e}{\tau} \big) = \big[\frac{V_q}{-V_e} \exp\big({\frac{V_e}{V_q}}\big) - E_1\big( \frac{-V_e}{V_q}\big) + \frac{V_e e + V_q \tau }{\tau V_e }\exp \big({ \frac{\tau V_e }{V_e e + V_q \tau }}\big)+E_1\big(\frac{- \tau V_e }{V_e e + V_q \tau } \big)\big]  \mathbf{1}\big(\frac{V_q}{-V_e}>\frac{e}{\tau} \big) +\big[\frac{V_q}{-V_e}\exp\big(\frac{V_e}{V_q} \big)- E_1\big( \frac{-V_e}{V_q}\big)\big]  \mathbf{1}\big(\frac{V_q}{-V_e}<\frac{e}{\tau} \big) \triangleq G\big( \frac{V_q}{-V_e}, \frac{e}{\tau}\big)$. Similarly, using the  integration by parts, we have  $\mathbb{E}\big[R(h,p^\ast)\big]= \big[\exp\big({\frac{\tau}{e}}\big) E_1\big(\frac{\tau^2 V_q }{e^2 V_e  + e\tau V_q }  \big)+E_1\big( \frac{-V_e}{V_q}\big)- E_1\big(\frac{- \tau V_e }{V_e e + V_q \tau } \big)\big]\mathbf{1}\big(\frac{V_q}{-V_e}>\frac{e}{\tau} \big) +E_1\big( \frac{-V_e}{V_q}\big)  \mathbf{1}\big(\frac{V_q}{-V_e}<\frac{e}{\tau} \big) \triangleq F\big( \frac{V_q}{-V_e}, \frac{e}{\tau}\big)$. Therefore, the PDE in (\ref{pdewe}) becomes:
\begin{align}
	\frac{q}{\overline{\lambda}\tau} &+   V_q\left[\overline{\lambda}-   F\left( \frac{V_q}{-V_e}, \frac{e}{\tau}\right)   \right] + V_e\left[\overline{\alpha}-G\left( \frac{V_q}{-V_e}, \frac{e}{\tau}\right) \right] =0	\label{pdenow11}
\end{align}

We then discuss the properties of $F$ and $G$ in (\ref{pdenow11}) as follows:
\begin{itemize}
	\item	If $\frac{V_q}{-V_e}\leq\frac{e}{\tau}$, $F$ is increasing w.r.t. $\frac{V_q}{-V_e}$ and $F\in [0, E_1 \left(\frac{\tau}{e} \right)]$. If $\frac{V_q}{-V_e}>\frac{e}{\tau}$, $F$ is a function of $\frac{V_q}{-V_e}$ and $\frac{e}{\tau}$, and is  increasing w.r.t.  $\frac{V_q}{-V_e}$ and $F\in (E_1 \left(\frac{\tau}{e} \right), \exp\left(\frac{\tau}{e} \right) E_1 \left( \frac{\tau}{e} \right))$.
	\item	If $\frac{V_q}{-V_e}\leq\frac{e}{\tau}$, $G$ is increasing w.r.t. $\frac{V_q}{-V_e}$ and $G\in [0, \frac{e}{\tau} \exp\left(-\frac{\tau}{e} \right)-E_1\left(\frac{\tau}{e} \right)]$. If $\frac{V_q}{-V_e}>\frac{e}{\tau}$, $G$ is a function of $\frac{V_q}{-V_e}$ and $\frac{e}{\tau}$, and is increasing w.r.t.  $\frac{V_q}{-V_e}$ and $G\in (\frac{e}{\tau} \exp\left(-\frac{\tau}{e} \right)-E_1\left(\frac{\tau}{e} \right), \frac{e}{\tau})$.
\end{itemize}

For the continuous time queueing system in (\ref{queue1}) and (\ref{queue2}), there exists a steady data queue states $q_s=0$ and $e_s\in[0,N_E]$, i.e., $\lim_{t \rightarrow \infty}q(t)=q_s$ and $\lim_{t \rightarrow \infty}e(t)=e_s$. At steady state, we require
\begin{align}
	&\overline{\lambda} \leq  F\left( \frac{V_q}{-V_e}, \frac{e_s}{\tau}\right),\qquad  \overline{\alpha} \geq G\left( \frac{V_q}{-V_e}, \frac{e_s}{\tau}\right)  \label{ssd2}
\end{align}
The existence of solution  for the HJB equation in  Theorem \ref{HJB11} is equivalent to the existence of solution of (\ref{ssd2}). We shall discuss the solution of (\ref{ssd2}) in the following two cases:

\emph{Case 1:} if the equalities are achieved in (\ref{ssd2}), i.e., 
\begin{align}
	&\overline{\lambda} =  F\left( \frac{V_q}{-V_e}, \frac{e_s}{\tau}\right),\qquad  \overline{\alpha} = G\left( \frac{V_q}{-V_e}, \frac{e_s}{\tau}\right)  \label{ssd2asd}
\end{align}
there exists a $\tilde{e}\in [0, N_E]$ such that
\begin{align}
	&\frac{\tilde{e}}{\tau} \exp\left(-\frac{\tau}{\tilde{e}} \right)-E_1\left(\frac{\tau}{\tilde{e}} \right) < \overline{\alpha} < \frac{\tilde{e}}{\tau}\notag \\
	&E_1 \left(\frac{\tau}{\tilde{e}} \right) < \overline{\lambda} < \exp\left(\frac{\tau}{\tilde{e}} \right) E_1 \left( \frac{\tau}{\tilde{e}} \right)	\label{satisfy2}
\end{align}
From the first equation above,  we have $\overline{\alpha}  \tau <\tilde{e} < x \tau$ where $x$ satisfies  $x\exp\left(-\frac{1}{x} \right)  - E_1 \left(\frac{1}{x} \right) = \overline{\alpha}$. For given $\tilde{e}$, we can obtain the range for $\overline{\lambda}$ according to the second equation above: $E_1\left(\frac{1}{\overline{\alpha}} \right) <  \overline{\lambda}  < \exp\left( \frac{1}{x}\right)E_1 \left(\frac{1}{x} \right)$. Furthermore,  we denote the solution of (\ref{ssd2asd}) w.r.t. $e$ to be $e^\ast$.  Then, it is sufficient that $N_E \geq e^\ast$  so that the  solution of (\ref{ssd2asd})  is meaningful. 

\emph{Case 2:} if $\overline{\lambda}\leq  E_1\left(\frac{1}{\overline{\alpha}} \right)$, we will show in Appendix I that the optimal control for this case achieves
\begin{align}
	\overline{\lambda} <  F\left( \frac{V_q}{-V_e}, \frac{e_s}{\tau}\right),\qquad  \overline{\alpha} = G\left( \frac{V_q}{-V_e}, \frac{e_s}{\tau}\right) 
\end{align}
where  the steady states are $q_s=0$, $e_s=\overline{\alpha}\tau$. In this case, we require that $N_E\geq e_s = \overline{\alpha}\tau$. Combining both cases, we obtain  the conditions in (\ref{energyrequir}) and (\ref{energyrequir1}). This completes the proof.

\section*{Appendix E: Proof of Theorem \ref{sym111}}
The PDE in (\ref{pdenow11}) has different structures when $\frac{V_q}{-V_e}<\frac{e}{\tau}$ and $\frac{V_q}{-V_e}>\frac{e}{\tau}$. Specifically, when  $\frac{V_q}{-V_e}>\frac{e}{\tau}$,
\begin{align}
	&\frac{q}{\overline{\lambda}\tau} +   V_q\left[\overline{\lambda}-   \exp\left({\frac{\tau}{e}}\right) E_1\left(\frac{\tau^2 V_q }{e^2 V_e  + e\tau V_q }  \right)-E_1\left( \frac{-V_e}{V_q}\right) \right. \notag \\
	&\left. +E_1\left(\frac{- \tau V_e }{V_e e + V_q \tau } \right)  \ \right] + V_e\left[\overline{\alpha}+\frac{V_q}{V_e} \exp\left({\frac{V_e}{V_q}}\right) +E_1\left( \frac{-V_e}{V_q}\right) \right. \notag \\
	&\left. -\frac{V_e e + V_q \tau }{\tau V_e }\exp \left({ \frac{\tau V_e }{V_e e + V_q \tau }}\right)-E_1\left(\frac{- \tau V_e }{V_e e + V_q \tau } \right)  \right] =0	\label{pdenow}
\end{align} when  $\frac{V_q}{-V_e}<\frac{e}{\tau}$,
\begin{align}
	&\frac{q}{\overline{\lambda}\tau}+   V_q\left[\overline{\lambda}-E_1\left( \frac{-V_e}{V_q}\right) \right]+ V_e\left[\overline{\alpha}+\frac{V_q}{V_e} \exp\left({\frac{V_e}{V_q}}\right)\right.\notag \\
	&\left. \hspace{5cm}+E_1\left( \frac{-V_e}{V_q}\right)  \right]=0 	\label{pdenow21}
\end{align}

\vspace{0.3cm}
\hspace{-0.3cm} \emph{A. Relationship among  $\frac{V_q}{-V_e}$, $\frac{e}{\tau}$, $\overline{\lambda}$ and $\overline{\alpha}$}

Dividing $-V_e$ on both sizes of (\ref{pdenow11}), we have
\begin{align}
	 &J\left( \frac{V_q}{-V_e}\right) \triangleq  \label{pdenow1} \\
	 &\frac{V_q}{-V_e}\left[\overline{\lambda}-    F\left( \frac{V_q}{-V_e}, \frac{e}{\tau}\right)   \ \right] - \left[\overline{\alpha}- G\left( \frac{V_q}{-V_e}, \frac{e}{\tau}\right)   \right] =- \frac{q}{-V_e\overline{\lambda}\tau}	\notag 
\end{align}
We first have the following lemma:
\begin{Lemma}	\label{lamdapro}
	From (\ref{pdenow1}), we have $-V_e=\Theta\left( \frac{1}{\overline{\lambda} \exp({\overline{\lambda}})}\right)$.~\hfill~\IEEEQED
\end{Lemma}
\begin{proof}	[Proof of Lemma \ref{lamdapro}]
	We assume that $V(q,e)=o(g(\overline{\lambda}))$ and  $V(q,e)=\mathcal{O}(f(\overline{\lambda}))$ for some functions $f$ and $g$. Therefore,  $V_q=o(g(\lambda))=\mathcal{O}(f(\lambda))$, and $V_e=o(g(\overline{\lambda}))=\mathcal{O}(f(\overline{\lambda}))$. According to (\ref{satisfy2}), we have  $\overline{\alpha}=o(\exp({\overline{\lambda}}))$.  Combining (\ref{pdenow11}), we have
	\begin{align}
		o(g(\overline{\lambda}))\Theta \left({\overline{\lambda}} \right)+o(g(\overline{\lambda}))\Theta(\exp({\overline{\lambda}}))=-\Theta \left(\frac{1}{\overline{\lambda}}\right)	\label{order1}	 \\
		\mathcal{O}(f(\overline{\lambda}))\Theta \left({\overline{\lambda}} \right)+\mathcal{O}(f(\overline{\lambda}))\Theta(\exp({\overline{\lambda}}))=-\Theta \left(\frac{1}{\overline{\lambda}}\right)			\label{order2}
	\end{align}
	 where (\ref{order1}) implies  $g(\lambda)=-o\left(\frac{1}{\overline{\lambda} \exp({\overline{\lambda}})}\right)$, and  (\ref{order2}) implies $f(\overline{\lambda})=-\mathcal{O}\left(\frac{1}{\overline{\lambda} \exp({\overline{\lambda}})}\right)$. Hence, $V(q,e)=-\Theta\left( \frac{1}{\overline{\lambda} \exp({\overline{\lambda}})}\right)$, which induces $-V_e=\Theta\left( \frac{1}{\overline{\lambda} \exp({\overline{\lambda}})}\right)$.\end{proof}

Based on Lemma \ref{lamdapro}, (\ref{pdenow1}) implies  $J\left( \frac{V_q}{-V_e}\right)=-\Theta(\exp(\overline{\lambda}))$.  Let $e^{th}$ satisfy $E_1 \left( \frac{\tau}{e^{th}} \right)=\overline{\lambda}$. We have the following discussions on the property of $J\left( \frac{V_q}{-V_e}\right)$:
\begin{itemize}
	\item  $e<e^{th}$:  if $\exp\left( \frac{\tau}{e^{th}}\right) E_1 \left(\frac{\tau}{e^{th} }\right) < \overline{\lambda}$, $J\left( \frac{V_q}{-V_e}\right)$ is an increasing function w.r.t. $\frac{V_q}{-V_e}$. Specifically, when $ 0<\frac{V_q}{-V_e}<x_0(e)$, where $F\left( x_0(e), \frac{e}{\tau}\right) =0$, $J\left( \frac{V_q}{-V_e}\right)$ is negative. When $ \frac{V_q}{-V_e}>x_0(e)$, $J\left( \frac{V_q}{-V_e}\right)$ is positive. On the other hand, if $\exp\left( \frac{\tau}{e^{th}}\right) E_1 \left(\frac{\tau}{e^{th} }\right) > \overline{\lambda}$, let $x_1(e)$ satisfy $F\left( x_1(e), \frac{e}{\tau}\right) =\overline{\lambda}$. When $ 0<\frac{V_q}{-V_e}<x_1(e)$,   $J\left( \frac{V_q}{-V_e}\right)$ is  increasing  w.r.t. $\frac{V_q}{-V_e}$, and when $\frac{V_q}{-V_e}>x_1(e)$,   $J\left( \frac{V_q}{-V_e}\right)$ is  decreasing  w.r.t. $\frac{V_q}{-V_e}$.
	
	\item  $e\geq e^{th}$: let $x_2(e)$ satisfy $ F\left( x_2(e), \frac{e}{\tau}\right) =\overline{\lambda}$. When $ 0<\frac{V_q}{-V_e}<x_2(e)$, $J\left( \frac{V_q}{-V_e}\right)$ is negative and increasing. When $ \frac{V_q}{-V_e}>x_2(e)$, $J\left( \frac{V_q}{-V_e}\right)$ is negative and decreasing. Furthermore, we have $x_0(e)<\frac{e}{\tau}$ for given $e$. 
\end{itemize}

Therefore, we  have the following results on the relationship among  $\frac{V_q}{-V_e}$, $\frac{e}{\tau}$, $\overline{\lambda}$ and $\overline{\alpha}$:
\begin{Classification}[Relationship among  $\frac{V_q}{-V_e}$, $\frac{e}{\tau}$, $\overline{\lambda}$ and $\overline{\alpha}$]	\label{classfic}	\
\begin{enumerate}	[1)]
	\item $e<e^{th}$:
		\begin{itemize}
			\item small $\overline{\lambda}$, small $\overline{\alpha}$ and $E_1\left(\frac{1}{\overline{\alpha}} \right)< \overline{\lambda} <\exp\left({\frac{1}{x}} \right) E_1\left(\frac{1}{x} \right)$:  in this case,  we have $\frac{V_q}{-V_e}=\Theta\left(\frac{\exp(\overline{\lambda})}{\overline{\lambda}} \right)$ which is large for sufficiently small $\overline{\lambda}$. Furthermore, since $e<e^{th}$, we have $\frac{V_q}{-V_e}>\frac{e}{\tau}$. Therefore,  the PDE in (\ref{pdenow1}) becomes (\ref{pdenow}) with large $ \frac{V_q}{-V_e}$ and small $e$.
			\item large $\overline{\lambda}$, large $\overline{\alpha}$ and $E_1\left(\frac{1}{\overline{\alpha}} \right)< \overline{\lambda} <\exp\left({\frac{1}{x}} \right) E_1\left(\frac{1}{x} \right)$:  similar to the previous case, we have $\frac{V_q}{-V_e}>\frac{e}{\tau}$. Since $e<e^{th}$ and we consider large $\overline{\lambda}$,  $e$ is relatively  large compared with $\frac{V_q}{-V_e}$. Therefore,  the PDE in (\ref{pdenow1}) becomes (\ref{pdenow}) with large $ \frac{V_q}{-V_e}$ and large $e$.
		\end{itemize}
	\item $e\geq e^{th}$: in this case, since $0<\frac{V_q}{-V_e}<x_0(e)$ and $x_0(e)<\frac{e}{\tau}$. Therefore, $\frac{V_q}{-V_e}<\frac{e}{\tau}$, which means that the PDE in (\ref{pdenow1}) becomes (\ref{pdenow21}).
\end{enumerate}
\end{Classification}

\vspace{0.3cm}
\hspace{-0.3cm} \emph{B. Solving the HJB Equation under the Large-Data Arrival-Energy-Sufficient  Regime}

According to Classification \ref{classfic}, when $e<e^{th}$, we have  the PDE in (\ref{pdenow}) with large $ \frac{V_q}{-V_e}$ and large $e$, and when $e\geq e^{th}$, we have the PDE in  (\ref{pdenow21}). We first solve the PDE in (\ref{pdenow}) with large $ \frac{V_q}{-V_e}$ and large $e$. We have the following approximations for $\frac{V_q}{-V_e}\mathbb{E}\left[R(h,p^\ast)\right] $ in (\ref{pdenow1}): $\frac{V_q}{-V_e}E_1\left(\frac{\tau^2 V_q }{e^2 V_e  + e\tau V_q }  \right)= \frac{V_q}{-V_e}E_1\left(\frac{\tau}{e} \right)+o(1)$, $E_1\left( \frac{-V_e}{V_q}\right)=-\gamma_{eu} \frac{V_q}{-V_e} +  \frac{V_q}{-V_e}\log\left(\frac{V_q}{-V_e} \right)+1 +o(1)$, $E_1\left(\frac{- \tau V_e }{V_e e + V_q \tau } \right)=-\gamma_{eu} \frac{V_q}{-V_e}+ \frac{V_q}{-V_e}\log\left(\frac{V_q}{-V_e}-\frac{e}{\tau} \right)+1 +o(1)$. Hence, we have
\begin{align}
	 \frac{V_q}{-V_e}\mathbb{E}\left[R(h,p^\ast)\right] &= \frac{V_q}{-V_e}\exp\left({\frac{\tau}{e}}\right) E_1\left(\frac{\tau}{e} \right)+ \frac{e}{\tau}+o(1)	\label{fieeq1}
\end{align}
Similarly, for $\mathbb{E}\left[p^\ast\right] $, we have  $\frac{V_q}{V_e} \exp\left({\frac{V_e}{V_q}}\right)=\frac{V_q}{V_e} + 1+o(1)$, $\frac{V_e e + V_q \tau }{\tau V_e }\exp \left({ \frac{\tau V_e }{V_e e + V_q \tau }}\right)=\frac{e}{\tau}+\frac{V_q}{V_e}+1+o(1)$, Hence, we have
\begin{align}
	\mathbb{E}\left[p^\ast\right]  = \frac{e}{\tau}	\label{fieeq2}
\end{align}
Substituting (\ref{fieeq1}) and (\ref{fieeq2}) into (\ref{pdenow11}) and for large $\overline{\lambda}$ and  $\overline{\alpha}$, we obtain the following simplified PDE:
\begin{align}
	\frac{q}{\overline{\lambda}\tau}+V_q \left(\overline{\lambda}-\exp\left({\frac{\tau}{e}}\right) E_1\left(\frac{\tau}{e} \right) \right) + V_e \overline{\alpha}=0	\label{extequasda}
\end{align}
For large $e$, we approximate $\exp\left({\frac{\tau}{e}}\right) E_1\left(\frac{\tau}{e} \right)$ as  $\exp\left({\frac{\tau}{e}}\right) E_1\left(\frac{\tau}{e} \right) = - \gamma_{eu}+\log \frac{e}{\tau}+o(1)$. Substituting  it  into (\ref{extequasda}) and using \emph{3.8.2.3} of \cite{pdebook}, we obtain $V(q,e)= \frac{e^2}{4 \overline{\lambda}\overline{\alpha}^2 \tau} \left(1+2\gamma_{eu}+2 \overline{\lambda} - 2 \log \frac{e}{\tau}\right)-\frac{eq}{\overline{\lambda} \overline{\alpha} \tau}+C$. We then determine the addend constant $C$.

For the steady state requirement in (\ref{ssd2}), using  (\ref{fieeq1}) and (\ref{fieeq2}), we have  $ -\gamma_{eu}+\log \frac{e}{\tau}+\frac{e}{\tau}\frac{-V_e}{V_q}=\overline{\lambda}$, $\frac{e}{\tau}=\overline{\alpha}$. We then obtain that $q_s=0$ $e_s=\overline{\alpha}\tau$, $\frac{V_q}{-V_e}\big|_{q=q_s, e=e_s}=\frac{\overline{\alpha}}{\overline{\lambda}+\gamma_{eu}-\log \overline{\alpha}}$. Under (\ref{finalsol1}), the  steady state requirement  is satisfied, and therefore the first condition in (\ref{trankern}) is satisfied. In addition, for any admissible $\Omega^v$, we have $\lim_{t\rightarrow \infty}q(t)=0$. Choosing $C=C_1$ as in (\ref{finalsol1}), the third condition in (\ref{trankern}) is satisfied  if  $N_E$ satisfies (\ref{energyrequir1}).  

We then solve the PDE in (\ref{pdenow21}) when $e\geq e^{th}$. To satisfy the second condition in (\ref{trankern}), it requires  $\frac{\partial V\left( q, N_E\right)}{\partial e}=0,  \forall q$. Using \emph{14.5.3.2} of \cite{pdebook}, we  obtain the solution  in the following form:
\begin{align}	\label{sol231sda}
	V(q,e)=c_1 e + \phi (q,c_1)+c_2
\end{align}
where  $\phi=\mathcal{O}(q^2)$. Then,  $\frac{\partial V\left( q, N_E\right)}{\partial e}=0$ induces  that $c_1=0$, which means that $V(q,e)$ is a function of $q$ only. Hence, $\frac{V_q}{-V_e}$ is infinite which means that $p^\ast=\frac{e}{\tau}$ according to Corollary \ref{corooptk}. Combing  (\ref{finalsol1}) ($e<e^{th}$) and (\ref{sol231sda}) ($e\geq e^{th}$), we obtain the full solution in this regime. 

\vspace{0.3cm}
\hspace{-0.3cm} \emph{C. Verification of the  Admissibility of $\Omega^{v \ast}$ under the Large-Data-Arrival-Energy-Sufficient  Regime}

We first calculate the water level when $e<e^{th}$ under (\ref{finalsol1}):
\begin{align}	\label{69ersa11}
	 \frac{V_q}{-V_e}=\frac{\overline{\alpha}e}{e\left(\gamma_{eu}+\overline{\lambda} - \log(\frac{e}{\tau}) \right)-\overline{\alpha}q}
\end{align}

Therefore, for sufficiently large $q$, we have $\frac{V_q}{-V_e}<0$,  which means that there is no data transmission, and the energy buffer will harvest energy until $e\geq e^{th}$ when the  policy is $p^\ast = \frac{e}{\tau}$ (we refer to it as the \emph{greedy policy}). Specifically, we can calculate  the trajectory of $e(t)$ as: $e(t) = (e(\bar{t})-\overline{\alpha} \tau) \exp(-t) + \overline{\alpha} \tau$, where $\bar{t}$ is the time stamp when  $e\geq e^{th}$ is first satisfied. Note that $\bar{t}=0$ if $e(0)\geq e^{th}$. This trajectory  implies that $\lim_{t \rightarrow \infty}e(t)=\overline{\alpha} \tau$.  For any $\epsilon>0$, there exists  $t_0>0$. When $t\geq t_0$, we have 
\begin{align}
	|e(t)-\overline{\alpha} \tau|\leq \epsilon, \quad t\geq t_0>\bar{t}	\label{epslondallan}
\end{align}

We  then can calculate the trajectory of $q(t)$ under $p^\ast=\frac{e}{\tau}$: $q(t) = q(0)+ q(\bar{t})-q(0) - \int_{\bar{t}}^t\big[ \exp\big(\frac{\tau}{e(t') } \big) E_1\big(\frac{\tau}{e(t') } \big)\\ -\overline{\lambda} \big] \tau t'+ \int_{\bar{t}}^t L(t)\mathrm{d}t=q(0)+ q(\bar{t})-q(0) - \int_{\bar{t}}^{t_0} \left[ \exp\big(\frac{\tau}{e(t') } \big) E_1\big(\frac{\tau}{e(t') } \big) -\overline{\lambda} \right] \tau t' +  \int_{\bar{t}}^{t_0} L(t)\mathrm{d}t 
 - \int_{t_0}^t \left[ \exp\big(\frac{\tau}{e(t') } \big) E_1\big(\frac{\tau}{e(t') } \big) \right. \\ \left. -\overline{\lambda} \right] \tau t'  + \int_{t_0}^{t} L(t)\mathrm{d}t$.	 Let $q(t_0)\triangleq q(\bar{t}) -q(0)- \int_{\bar{t}}^{t_0} \left[ \exp\left(\frac{\tau}{e(t') } \right) E_1\left(\frac{\tau}{e(t') } \right)-\overline{\lambda} \right] \tau t' +  \int_{\bar{t}}^{t_0} L(t)\mathrm{d}t $. Therefore, 
\begin{align}
	q(t)= & q(0)+ q(t_0) - \int_{t_0}^t \left[ \exp\left(\frac{\tau}{e(t') } \right) E_1\left(\frac{\tau}{e(t') } \right)-\overline{\lambda} \right] \tau t'\notag \\
	& + \int_{t_0}^{t} L(t')\mathrm{d}t' 	\notag \\
	      \overset{(a)}{\leq}& q(0)+ q(t_0) - \int_{t_0}^t \bigg[ \exp\left(\frac{1}{\overline{\alpha}  + \epsilon/\tau } \right) E_1\left(\frac{1}{\overline{\alpha} +\epsilon/\tau } \right)  \notag \\
	      & -\overline{\lambda} \bigg] \tau t' + \int_{t_0}^{t} L(t')\mathrm{d}t' 	\label{epsilondelads}
\end{align}
where $(a)$ is due to $e(t)<\overline{\alpha} \tau + \epsilon$ when $t\geq t_0$ according to (\ref{epslondallan}). Since $E_1\left(\frac{1}{\overline{\alpha}} \right) <  \overline{\lambda}$, there exists a  $\delta>0$ such that  $\overline{\lambda}< \exp\left(\frac{1}{\overline{\alpha}  + \delta } \right) E_1\left(\frac{1}{\overline{\alpha} + \delta } \right)$.	 Choosing $\epsilon=\delta \tau$ in (\ref{epsilondelads}), we obtain
\begin{align}
	q(t)-q(0)<0, \quad \text{if }	t\geq \frac{q(t_0) }{\exp\left(\frac{1}{\overline{\alpha}  + \delta } \right) E_1\left(\frac{1}{\overline{\alpha} + \delta } \right)-\overline{\lambda}}
\end{align}
Therefore, we obtain  the negative queue drift, which means that the greedy policy is a stabilizing policy \cite{stable1}, \cite{stable2}.

\section*{Appendix F: Proof of Corollary \ref{cor1sd111}}
Since $V(q,e)$ is a function of $q$ only when $e\geq e^{th}$ and hence, $\frac{V_q}{-V_e}$ is infinite, which means that $p^\ast=\frac{e}{\tau}$ according to Corollary \ref{corooptk}. For $e<e^{th}$, the water level (WL) is given in (\ref{69ersa11}). When $q>\frac{e}{\overline{\alpha}}\left(\gamma_{eu}+\overline{\lambda} - \log(\frac{e}{\tau}) \right)$, the WL is negative, which results in $p^\ast=0$.  On the other hand, when $q<\frac{e}{\overline{\alpha}}\left(\gamma_{eu}+\overline{\lambda} - \log(\frac{e}{\tau}) \right)$, the WL is positive and  increasing w.r.t.  $q$. Moreover,  $\text{derivative of (\ref{69ersa11}) w.r.t. $e$}=\frac{\overline{\alpha} (e - \overline{\alpha}  q)}{(-e(\overline{\lambda}+\gamma_{eu})+\overline{\alpha} q + e \log(\frac{e}{\tau}))^2}$. When $e<\overline{\alpha}  q$, the WL is  decreasing  w.r.t. $e$, and when $e<\overline{\alpha}  q$,  the WL is  increasing  w.r.t.  $e$.

\section*{Appendix G: Proof of Theorem \ref{sym1}}
\vspace{0.3cm}
\hspace{-0.3cm} \emph{A. Solving the HJB Equation under the Small-Data-Arrival-Energy-Limited  Regime}

According to Classification \ref{classfic}, when $e<e^{th}$, we have  the PDE in (\ref{pdenow}) with large $ \frac{V_q}{-V_e}$ and small $e$ and when $e\geq e^{th}$, we have  the PDE in (\ref{pdenow21}). Following part B in Appendix E, we can obtain the simplified PDE as in (\ref{extequasda}). For small $e$, we approximate $\exp\left({\frac{\tau}{e}}\right) E_1\left(\frac{\tau}{e} \right)$ as $\exp\left({\frac{\tau}{e}}\right) E_1\left(\frac{\tau}{e} \right) = \frac{e}{\tau}+o(1)$. Substituting it to (\ref{extequasda}) and using \emph{3.8.2.3} of \cite{pdebook}, we obtain the  solution for this case as in (\ref{finalsol2}). Furthermore, the solution for $e\geq e^{th}$ is the same as (\ref{sol231sda}). Following the same procedure in Appendix E,  it can be verified that the three conditions in (\ref{trankern}) are satisfied.

\vspace{0.3cm}
\hspace{-0.3cm} \emph{B. Verification of the  Admissibility of $\Omega^{v \ast}$ under the Small-Data-Arrival-Energy-Limited  Regime}

We first calculate the WL when $e<e^{th}$ under (\ref{finalsol2}):
\begin{align}	\label{69ersa1}
	 \frac{V_q}{-V_e}=\frac{\overline{\alpha}\tau e}{-e^2+\overline{\lambda}\tau e - \overline{\alpha}\tau q}
\end{align}
Therefore, for sufficiently large $q$, we have $\frac{V_q}{-V_e}<0$,  which means that there is no data transmission, and the energy buffer will harvest energy until $e\geq e^{th}$ when the data queue will adopt the  policy $p^\ast = \frac{e}{\tau}$.  Following the same proof as in part C in Appendix E, we can  prove  the negative data queue drift.

\section*{Appendix H: Proof of Corollary \ref{cor1wewe111}}
Since $V(q,e)$ is a function of $q$ only when $e\geq e^{th}$, we have $p^\ast=\frac{e}{\tau}$. For $e<e^{th}$, the WL is given in (\ref{69ersa1}). When $q>\frac{-e^2 + \overline{\lambda} \tau e}{\overline{\alpha}\tau}$, the WL is negative, which results in $p^\ast=0$.  On the other hand, when $q<\frac{\overline{\lambda} e+\overline{\alpha}e }{\overline{\alpha} }$, the WL is positive, which is  increasing w.r.t.  $q$. Moreover, $\text{derivative of (\ref{69ersa1}) w.r.t. $e$}=\frac{\overline{\alpha}\tau (e^2 - \overline{\alpha} \tau q)}{(-e^2+\overline{\lambda}\tau e - \overline{\alpha}\tau q)^2}$. When $e<\sqrt{\overline{\alpha} \tau q}$, the WL is  decreasing  w.r.t. $e$, and when $e>\sqrt{\overline{\alpha} \tau q}$,  the WL is  increasing  w.r.t.  $e$.

\section*{Appendix I: Proof of Theorem \ref{sym33}}

\vspace{0.3cm}
\hspace{-0.3cm} \emph{A. Solving the HJB Equation under the Small-Data-Arrival-Energy-Sufficient  Regime}

Following the same analysis as in part A in Appendix E, when $e<e^{th}$, we have the PDE  in (\ref{pdenow}), and when $e\geq e^{th}$, we have  the  PDE in (\ref{pdenow21}).  For the PDE in (\ref{pdenow}), we require
\begin{align}
	\frac{\partial V\left( 0, e\right)}{\partial q}=0	\label{bound2}
\end{align}
because  the equalities in (\ref{ssd2})  cannot be achieved and $L(t)\neq 0$ after the virtual queueing system enters the steady state. Under this regime, the system operates at the region with small  ${V_q}$. We have the following approximations for $\mathbb{E}\left[R(h,p^\ast)\right] $ in (\ref{pdenow}): $\exp\left({\frac{\tau}{e}}\right)E_1\left(\frac{\tau^2 V_q }{e^2 V_e  + e\tau V_q }  \right)= \frac{e}{\tau} \left(1-\frac{e}{\tau}\frac{-V_e}{V_q}\right)  \exp\left({- \frac{1}{\frac{V_q}{-V_e}-\frac{e}{\tau}}} \right)+ o(1)$, $E_1\left( \frac{-V_e}{V_q}\right)=\frac{V_q}{-V_e} \exp\left({\frac{V_e}{V_q}}\right)+o(1)$, $E_1\left(\frac{- \tau V_e }{V_e e + V_q \tau } \right)=\left(\frac{V_q}{-V_e}-\frac{e}{\tau}\right)  \exp\left({- \frac{1}{\frac{V_q}{-V_e}-\frac{e}{\tau}}} \right)+ o(1)$. Hence, we have
\begin{align}
	 \mathbb{E}\left[R(h,p^\ast)\right] &= \mathcal{O}\left(\frac{V_q}{-V_e}\exp\left({\frac{V_e}{V_q}}\right) \right)+o(1)=o(1)	\label{fieeq12}
\end{align}
Similarly, for $\mathbb{E}\left[p^\ast\right] $, we have
\begin{align}
	\mathbb{E}\left[p^\ast\right]  = o(1)	\label{fieeq22}
\end{align}
Substituting (\ref{fieeq12}) and (\ref{fieeq22}) into (\ref{pdenow11}), we obtain the  simplified PDE: $\frac{q}{\overline{\lambda}\tau}+V_q \overline{\lambda}+ V_e \overline{\alpha}=0$. Using \emph{3.8.2.3} of \cite{pdebook}, we obtain the following solution for this case:
\begin{align}	\label{sol221}
	V(q,e)=\frac{1}{2 \overline{\alpha}^2 \tau} e^2 - \frac{e q}{\overline{\lambda}\overline{\alpha}\tau}+ \phi\left({q}-\frac{\overline{\lambda}}{\overline{\alpha}}e\right)
\end{align}
From (\ref{bound2}), we require that $\phi'\left( -\frac{\overline{\lambda}}{\overline{\alpha}}e\right)=\left( -\frac{\overline{\lambda}}{\overline{\alpha}}e \right)\left(-\frac{1}{\lambda^2\tau}\right)$, $\forall e$. We choose $\phi(x)=x^2\left(-\frac{1}{2\lambda^2\tau}\right)$. Therefore, the final solution is given in (\ref{finalsol3}).  Furthermore, the solution for $e\geq e^{th}$ is the same as (\ref{sol231sda}). Following the same procedure  in Appendix E,  it can be verified that the three conditions in (\ref{trankern}) are satisfied.

\vspace{0.3cm}
\hspace{-0.3cm} \emph{B. Verification of the  Admissibility of $\Omega^{v \ast}$ under the Small-Data-Arrival-Energy-Sufficient  Regime}

Note that when $e<e^{th}$, under the solution in (\ref{finalsol3}), we have that $\frac{\partial V(q,e)}{\partial e}=0$ for all $q, e$, which results in $p^\ast=\frac{e}{\tau}$.  Following the same proof as in part C in Appendix E, we can  prove  the negative data queue drift.

\section*{Appendix J: Proof of Theorem \ref{optimawlcontrolreg11}}
We  prove that for sufficiently large queue $Q(0)$, for the following case 1 ($E(0)>e^{th}$) and case 2 ($E(0)<e^{th}$), we have negative data queue drift.

\emph{Case 1, $E(0)>e^{th}$:} In this case, the greedy  policy $p^\ast(n)=\frac{E(n)}{\tau}$ is adopted for all different asymptotic scenarios. Based on the energy queue dynamics in (\ref{energyQ}), we have $p^\ast(n)=\min\{\alpha(n-1), \frac{N_E}{\tau}\}$ for $n \geq 1$. We then calculate the one step queue drift as follows: for sufficiently large $Q$,
\begin{align}
	&\mathbb{E}\left[Q(n+1)-Q(n)\big|Q(n)=Q, E(n)=E \right]	\notag \\
	=&\mathbb{E}\left[\left[ Q -  \log\left(1+|h|^2 E/\tau \right) \tau  \right]^+ + \lambda  \tau - Q\right]\notag\\
	\overset{(a)}{=}&\mathbb{E}\left[-  \log\left(1+|h|^2 E/\tau \right) \tau   + \lambda  \tau \right]	\label{mmsssg}
\end{align}
where (a) is due to the fact that for  given $E$ and sufficiently large $Q$, we have $\Pr\big[ Q >  \log\big(1+|h|^2 E \big) \tau \big]=\Pr\big[|h|^2<\frac{\exp(Q)-1}{e/\tau} \big]>1-\delta$ ($\forall \delta>0$). In (\ref{mmsssg}),  if $\alpha(n-1)<\frac{N_E}{\tau}$, we have $\mathbb{E}\big[-  \log\big(1+|h|^2 E/\tau \big) \tau   + \lambda  \tau \big]=\mathbb{E}\big[-  \log\big(1+|h|^2 \alpha \big) \tau   + \lambda  \tau \big]=\big(\overline{\lambda}-\mathbb{E}\big[\exp\big({\frac{1}{{\alpha}}} \big)E_1\big(\frac{1}{{\alpha}} \big)\big]\big)\tau\overset{(b)}{<}0$, where (b) is due to (\ref{energyrequqweir}). If $\alpha(n-1)>N_E$, we have  $\mathbb{E}\big[-  \log\big(1+|h|^2 E/\tau \big) \tau   + \lambda  \tau \big]	 =\mathbb{E}\big[-  \log\big(1+|h|^2 N_E/\tau \big) \tau   + \lambda  \tau \big] \overset{(c)}{\leq} \mathbb{E}\big[-  \log\big(1+|h|^2 \overline{\alpha} \big) \tau   + \lambda  \tau \big] =\big(\overline{\lambda}-\exp\big({\frac{1}{\overline{\alpha}}} \big)E_1\big(\frac{1}{\overline{\alpha}} \big)\big)\tau\overset{(d)}{<}0$, where (c) is due to $N_E\geq N\overline{\alpha}\tau>\overline{\alpha}\tau$ and (d) is due to $\overline{\alpha}\leq \mathbb{E}\big[\exp\big(\frac{1}{\alpha} \big)E_1\big(\frac{1}{\alpha} \big)\big]<\exp\big({\frac{1}{\overline{\alpha}}} \big)E_1\big(\frac{1}{\overline{\alpha}} \big)$. Hence,  we have negative drift.

\emph{Case 2, $E(0)<e^{th}$:} In this case, we  show that there exists some positive integer $n<N$ such that the $n$-step queue drift in the discrete time queueing system is negative.  Since $E(0)<e^{th}$ and $Q(0)$ is sufficiently large, the data queue will not transmit  in the beginning. For given $\alpha$, after $\lceil \frac{e^{th}-E(0)}{\alpha} \rceil$ number of time slots where $\lceil x \rceil$ is the ceiling function, the data queue will adopt the greedy policy to transmit. To prove the existence of $n$, it is sufficient to prove that 
\begin{align}
	\mathbb{E}\left[ \left(N - \lceil \frac{e^{th}-E(0)}{\alpha} \rceil \right) \left(\exp\left(\frac{1}{\alpha}\right)E_1\left(\frac{1}{\alpha} \right)  \right)\right]>\mathbb{E}\left[ \lambda N\right]	\label{lhsrhsexp}
\end{align}	
where the L.H.S. (R.H.S.) means the departure   bits (arrival bits) before the end of the next change event of the energy arrival rate. From (\ref{lhsrhsexp}), we have $(\ref{lhsrhsexp}) \Leftarrow \mathbb{E}\left[ \left(1 - \frac{1}{N}\lceil \frac{e^{th}-E(0)}{ \alpha} \rceil \right) \left(\exp\left(\frac{1}{\alpha}\right)E_1\left(\frac{1}{\alpha} \right)  \right)\right]>\overline{ \lambda}	\overset{(e)}\Leftarrow \mathbb{E}  \left(\exp\left(\frac{1}{\alpha}\right)E_1\left(\frac{1}{\alpha} \right)  \right)	\overset{(f)}  >\overline{ \lambda}$, where $(e)$ holds for  large $N$ and $(f)$ holds due to (\ref{energyrequqweir}). Therefore,  we have negative drift for this case. Based on the Lyapunov theory \cite{stable1}, \cite{stable2}, negative state drift for both cases  leads to the stability of $Q(n)$, i.e., $\lim_{n \rightarrow \infty} \mathbb{E} \big[  Q^2(n)  \big] < \infty$.

\end{document}